\documentclass{LMCS}

\def\dOi{10(1:10)2014}
\lmcsheading%
{\dOi}
{1--33}
{}
{}
{May.~22, 2013}
{Feb.~13, 2014}
{}

\ACMCCS{[{\bf Theory of computation}]: Formal languages and automata theory}


\usepackage{hyperref}
\usepackage{amssymb}

\usepackage{dsfont}
\usepackage{xspace}

\usepackage{setspace}

\usepackage{color}

\usepackage{multirow}

\usepackage{pifont}

\newcommand{\cal}{\mathcal}
\newcommand{\con}{\!\cdot\!}
\newcommand{\lcon}{\cdot}

\newcommand{\Nat}{\mathds{N}}
\newcommand{\Rat}{\mathds{Q}}
\newcommand{\Reals}{\mathds{R}}

\newcommand{\A}{{\mathcal A}}
\newcommand{\B}{{\mathcal B}}
\newcommand{\C}{{\mathcal C}}
\newcommand{\D}{{\mathcal D}}
\newcommand{\G}{{\mathcal G}}

\newcommand{\DSA}{DA}
\newcommand{\NDA}{N\DSA\xspace}
\newcommand{\NDAs}{N{\DSA}s\xspace}
\newcommand{\DDA}{D\DSA\xspace}
\newcommand{\DDAs}{D{\DSA}s\xspace}

\newcommand{\Func}[1]{\mathtt{#1}}
\newcommand{\range}{\Func{range}}
\newcommand{\Gap}{\Func{gap}}
\newcommand{\Cost}{\Func{cost}}
\newcommand{\Round}{\Func{Round}}

\newcommand{\letter}[1]{\mbox{`}#1\mbox{'}}

\def\abs#1{\ensuremath{\lvert #1 \rvert}}
\newcommand{\lwhole}[1]{\lfloor #1 \rfloor}
\newcommand{\uwhole}[1]{\lceil #1 \rceil}
\newcommand{\half}{\frac{1}{2}}

\newcommand{\tuple}[1]{\langle #1  \rangle}
\newcommand{\pair}{\tuple}

\newcommand{\ST}{~{\big |}\:}
\newcommand{\emptyword}{\varepsilon}

\newcommand{\NiceCheckMark}{\ding{51}}
\newcommand{\NiceCross}{\ding{55}}

\newcommand{\Paragraph}[1]{\subsection*{#1}}
\begin{document}

\title[Exact and Approximate Determinization of Discounted-Sum Automata]
      {Exact and Approximate Determinization of Discounted-Sum Automata\rsuper*}
\author[U.~Boker]{Udi Boker\rsuper a}	
\address{{\lsuper a}The Interdisciplinary Center, Herzliya, Israel}	

\author[T.~A.~Henzinger]{Thomas A. Henzinger\rsuper b}	
\address{{\lsuper b}IST Austria, Klosterneuburg, Austria}	
\thanks{{\lsuper b}This work was supported in part by the Austrian
  Science Fund NFN RiSE (Rigorous Systems Engineering) and by the ERC
  Advanced Grant QUAREM (Quantitative Reactive Modeling).}	

\titlecomment{{\lsuper*}The present article combines and extends \cite{BH11,BH12}.}
%%%%%%%%%%%%%%%%%%%%%%%%%%%%%%%%%%%%%%%%%%%%%%%%%%%%%%%%%%%%%%%%%%%%%%%%%%%

%% the abstract has to PRECEED the command \maketitle:
%% be sure not to issue the \maketitle command twice!

\begin{abstract}
\noindent A discounted-sum automaton (\NDA) is a nondeterministic finite automaton 
with edge weights, valuing a run by the discounted sum of visited 
edge weights.
More precisely, the weight in the $i$-th position of the run is divided by 
$\lambda^i$, where the discount factor $\lambda$ is a fixed rational number 
greater than~$1$. The value of a word is the minimal value of the automaton runs on it. 
Discounted summation is a common and useful measuring scheme, especially 
for infinite sequences, reflecting the assumption that earlier weights 
are more important than later weights. Unfortunately, determinization of \NDAs, which is often essential in formal verification, is, in general, not possible.

We provide positive news, showing that every \NDA with an integral discount factor is determinizable. 
We complete the picture by proving that the integers characterize exactly the discount factors that guarantee determinizability:
for every rational discount factor $\lambda\not\in\Nat$, there is a nondeterminizable $\lambda$-\NDA. 
We also prove that the class of \NDAs with integral discount factors enjoys closure under the algebraic operations $\min$, $\max$, addition, and 
subtraction, which is not the case for general \NDAs nor for deterministic \NDAs.

For general \NDAs, we look into approximate determinization, which is always possible as the influence of a word's suffix decays. We show that the naive approach, of unfolding the automaton computations up to a sufficient level, is doubly exponential in the discount factor. We provide an alternative construction for approximate determinization, which is singly exponential in the discount factor, in the precision, and in the number of states. We also prove matching lower bounds, showing that the exponential dependency on each of these three parameters cannot be avoided.

 All our results hold equally for automata over finite words and for automata over infinite words.
\end{abstract}

\maketitle
\vfill
\section{Introduction}\label{sec:Introduction}

Discounting the influence of future events is a key paradigm in economics and it is studied in game theory (e.g.\   \cite{ZP96,Andersson06}), Markov decision processes (e.g.\ \cite{DiscountedMarkov,DiscountedDeterministicMarkov}), and automata theory (e.g.\ \cite{DiscountingInSystems,Skew,AlternatingWeightedAutomata,ExpressivenessQuantitativeLanguages,CDH10}). Discounted summation formalizes the concept that an immediate reward is better than a potential one in the far-away future, as well as that a potential problem in the future is less troubling than a current one.

A discounted-sum automaton (\NDA) is a nondeterministic automaton with rational weights on the transitions, where the value of a run is the discounted summation of the weights along it. Each automaton has a fixed discount-factor $\lambda$, which is a rational number bigger than $1$, and the weight in the $i$th position of a run is divided by $\lambda^i$. The value of a word is the minimal value of the automaton runs on it. Hence, an \NDA realizes a function from words to real numbers.  Two automata are equivalent if they realize the same function, namely if they assign the same value to every word. 

Discounted summation is of special interest for automata over infinite words. There are two common ways to adjust standard summation for handling infinite sequences: discounting and limit-averaging. The latter, which relates to the input suffixes, has been studied a lot in mean-payoff games and, more recently, in limit-average automata \cite{CDH10, ImperfectInformation}; the former, which relates more to the input prefixes, has received comparatively little attention.

Automata are widely used in formal verification, for which automata comparison is fundamental. Specifically, one usually considers the following three questions, ordered from the most difficult one to the simplest one: general comparison (language inclusion), universality, and emptiness. In the Boolean setting, where automata assign Boolean values to the input words, the three questions, with respect to automata $\A$ and $\B$, are whether $\A\subseteq\B$, $\A=True$, and $\A=False$. In the quantitative setting, where automata assign numeric values to the input words, the universality and emptiness questions relate to a constant threshold, usually $0$. Thus, the three questions are whether $\A\leq\B$, $\A\leq 0$, and $\A\geq 0$.

A central problem with these quantitative automata is that only the emptiness question is known to be solvable \cite{CDH10}. For limit-average automata, the two other questions are undecidable \cite{ImperfectInformation}. For \NDAs, it is an open question whether universality and comparison are decidable. (For special cases, such as ``functional automata'', where all runs over a single word yield the same value, the problem is decidable \cite{FGR12}.) This is not the case with \DDAs, for which all three questions have polynomial solutions \cite{ZP96,Andersson06,CDH10}. Unfortunately, \NDAs cannot, in general, be determinized. It is currently known that for every rational discount-factor $1 < \lambda < 2$,
there is a $\lambda$-\NDA that cannot be determinized \cite{CDH10}. 

It turns out, quite surprisingly, that discounting by an integral factor forms a ``well behaved'' class of automata, denoted ``integral \NDAs'', allowing for determinization (Section~\ref{sec:Determinizability}) and closed under the algebraic operations $\min$, $\max$, addition and subtraction (Section~\ref{sec:Closure}). The above closure is of special interest, as neither \NDAs nor \DDAs are closed under the $\max$ operation (Theorem~\ref{thm:MaxInclosure}). Furthermore, the integers, above $1$, characterize exactly the set of discount factors that guarantee determinizability (Section~\ref{sec:NonDeterminizability}). That is, for every rational factor $\lambda\not\in\Nat$, there is a non-determinizable $\lambda$-\NDA. 

The discounted summation intuitively makes \NDAs more influenced by word-prefixes than by word-suffixes, suggesting that some basic properties are shared between automata over finite words and over infinite words. Indeed, all the above results hold for both models. Yet, the equivalence relation between automata over infinite words is looser than the one on finite words. That is, if two automata are equivalent with respect to finite words then they are also equivalent with respect to infinite words, but not vice versa (Lemma~\ref{lem:FiniteEquivalenceImpliesInfiniteEquivalence}). 

The above results relate to complete automata; namely, to automata in which every state has at least one transition over every alphabet letter. Also, our automata do not have a Boolean acceptance condition, which would have made them compute a partial function instead of a total one. For incomplete automata or, equivalently, for automata with $\infty$-weights, or automata where some of the states are accepting and some are not, no discount factor can guarantee determinization (Section~\ref{sec:Incomplete}). In the scope of discounted-sum automata, the restriction to complete automata is very natural, as infinite-weight edges break the property of the decaying importance of future events.

Our determinization procedure, described in Section~\ref{sec:Construction}, is an extension of the subset construction, keeping a ``recoverable-gap'' value to each element of the subset. Intuitively, the ``gap'' of a state $q$ over a finite word $u$ stands for the extra cost of reaching $q$, compared to the best possible value so far. This extra cost is multiplied, however, by $\lambda ^{|u|}$, to reflect
the $\lambda ^{|u|}$ division in the value-computation of the suffixes. A gap of $q$ over $u$ is ``recoverable'' if there is a suffix $w$ that ``recovers'' it, meaning that there is an optimal run over $uw$ that visits $q$ after reading $u$. Due to the discounting of the future, once a gap is too large, it is obviously not recoverable. Specifically, for every $\lambda$, we have that $\sum_{i=0}^\infty (\frac{1}{\lambda^i}) = \frac{1}{1-\frac{1}{\lambda}} =
\frac{\lambda}{\lambda-1} \leq 2$. Hence, our procedure only keeps gaps that are smaller than twice the maximal difference between the automaton weights. 

The determinization procedure may be used for an arbitrary $\lambda$-\NDA, always providing an equivalent $\lambda$-\DDA, if terminating. Yet, it is guaranteed to terminate for a $\lambda$-\NDA with $\lambda\in\Nat$, while it might not terminate in the case that $\lambda\in\Rat\setminus\Nat$. 

For integral \NDAs, the key observation is that there might only be finitely many recoverable gaps (Lemma~\ref{lem:Termination}). More precisely, for an integral  \NDA $\A$, there might be up to $m$ recoverable gaps, where $m$ is the maximal difference between the weights in $\A$, multiplied by the least common denominator of all weights. Accordingly, our determinization procedure generates a \DDA with up to $m^n$ states, where $n$ is the number of states in $\A$. We show, in Section~\ref{sec:StateComplexity}, that there must indeed be a linear dependency on the weight values (and thus possibly an exponential dependency on their representation), as well as an exponential dependency on the number of states. 
%As for the exact dependency of the state blowup on the combination of the weights and the number of states, we show a tight lower bound of $O(m)^{O(n)}$, using a rich alphabet of size exponential in the number of states; The exact unavoidable state blow-up for the case that the alphabet size is linear in the number of states is left as an open problem.

For nonintegral \NDAs, the key observation is that the recoverable gaps might be arbitrarily close to each other (Theorem~\ref{thm:NotNaturalNumber}). Hence, the bound on the maximal value of the gaps cannot guarantee a finite set of recoverable gaps. Different gaps have, under the appropriate setting, suffixes that distinguish between them, implying that an equivalent deterministic automaton must have a unique state for each recoverable-gap (Lemma~\ref{lem:DifferentGaps}). Therefore, an automaton that admits infinitely many recoverable gaps cannot be determinized.

As nonintegral \NDAs cannot, in general, be determinized, we investigate, in Section~\ref{sec:ApproxDet}, their approximate determinization. We define that an automaton can be determinized approximately if for every real precision $\varepsilon>0$, there is a deterministic automaton such that the difference between their values on all words is less than or equal to $\varepsilon$. Due to the discounting in the summation, all \NDAs allow for approximate determinization, by unfolding the automaton computations up to a sufficient level. This is in contradistinction to other common quantitative automata, such as sum, average and limit-average automata, which cannot be determinized approximately \cite{BH12}.
The smaller the required precision is, and the closer the discount-factor is to $1$, the more expensive it is to determinize an \NDA approximately. We represent the precision by $\varepsilon=2^{-p}$ and the discount factor by $\lambda=1+2^{-k}$. We analyze the unfolding approach to construct an automaton whose state space is exponential in the representation of the precision ($p$) and doubly exponential in the representation of the discount factor ($k$).

We provide an alternative construction for approximate determinization, by generalizing the determinization procedure of Section~\ref{sec:Construction}. Recall that this procedure need not terminate for a nonintegral discount factor, since there might be infinitely many recoverable gaps. We overcome the problem by rounding the gaps to a fixed resolution.
This obviously guarantees termination, however it raises the question of how an unbounded number of gap rounding allows for the required precision. The key observation is that the rounding is also discounted along the computation. We show that the construction is singly exponential in $k$, in $p$, and in the number of states of the automaton. We complete the picture by proving matching lower bounds, showing exponential dependency on each of these three parameters.

It turns out that closure under algebraic operation is also closely related to the question of whether the set of recoverable gaps is finite. Considering the operations of addition, subtraction, minimum, and maximum, between two automata, the latter is the most problematic one, as the value of a word is defined to be the minimal value of the automaton runs on it. For two \NDAs, $\A$ and $\B$, one may try to construct an automaton $\C=\max(\A,\B)$, by taking the product of $\A$ and $\B$, while maintaining the recoverable gaps of $\A$'s original states, compared to $\B$'s original states. This approach indeed works for integral \NDAs (Theorem~\ref{thm:MaxClosure}). Note that determinizability is not enough, as neither \NDAs nor \DDAs are closed under the $\max$ operation. Furthermore, we show, in Theorem~\ref{thm:MaxInclosure}, that there are two \DDAs, $\A$ and $\B$, such that there is no \NDA $\C$ with $\C=\max(\A,\B)$. For precluding the existence of such a nondeterministic automaton $\C$, we cannot make usage of Lemma~\ref{lem:DifferentGaps}, and thus use a more involved, ``pumping-style'', argument with respect to recoverable gaps.   

\Paragraph{Related work} Weighted automata are often handled as \emph{formal power series}, mapping words to a \emph{semiring} \cite{WeightedHandbook}. By this view, the weight of a run is the semiring-multiplication of the transition weights along it, while the weight of a word is the semiring-addition of its possible run weights. Discounted-summation is, upfront, not an associative operation. Yet, it can be encoded as such, allowing to view it as operating over a semiring. For the semiring setting, there are numerous works, including results on determinization \cite{MohriAlgorithms,WeightedHandbook}. 
Nevertheless, the algorithms for determinizing arbitrary automata over semirings are general, and do not take advantage of the special properties of discounted summation. For that reason, our algorithm is guaranteed to terminate over every integral discounted-sum automaton, which is not the case for the general determinization algorithm.
In particular, our determinization algorithm differs from the algorithm in \cite[Chapter 7.2]{WeightedHandbook} in two main aspects that are special to discounted summation: (a) the data associated with the states of the deterministic automata concern, in our algorithm, \emph{gaps}, whereas in \cite{WeightedHandbook} it concerns \emph{residuals}. A residual of a (finite) path is the difference between the accumulated weights along it and the accumulated weights along the (so far) optimal path. A gap is the ``future difference'' between these accumulated values, meaning the extra cost that will equalize the two paths, taking into account the future discounting; and (b) Due to the future discounting, gaps have a maximal relevant value (threshold), over which they can be disregarded, which is not the case with residuals.

Formal power series are also generalized, in \cite{Skew}, for handling infinite discounted summation. The weight of a run is defined to be a ``skewed multiplication'' of the weights along it, where this ``skewing'' corresponds to the discounting operation. Yet, \cite{Skew} mainly considers the equivalence between recognizable series and rational series, and does not handle automata determinization.

Some determinization algorithms of weighted automata are guaranteed to terminate, provided that the nondeterministic automaton satisfies the ``twins'' property \cite{MohriAlgorithms,WeightedHandbook,AKL11}. Roughly speaking, the twins property says that different runs over the same input provide the same value, up to a finite set of differences between their values. These differences in the values stem from the different prefixes (delays) that might lead to different loops over the same input. These loops, however, must yield the same value. For that reason, the twins property is a sufficient condition for the termination of the general determinization algorithm of \cite{WeightedHandbook}. Our notion of gaps is unrelated to the twins property, and does not depend on a specific structure of the nondeterministic automaton. It allows the determinization algorithm to terminate over every integral discounted-sum automaton, whether or not it has the twins property.

Discounted Markov decision processes (e.g.\ \cite{DiscountedMarkov,DiscountedDeterministicMarkov}) and discounted games (e.g.\   \cite{ZP96,Andersson06}) generalize, in some sense, deterministic discounted-sum automata. The former adds probabilities and the latter allows for two player choices. However, they do not cover nondeterministic automata. One may note that nondeterminism relates to ``blind games'', in which each player cannot see the other player's moves, whereas in standard games the players have full information on all moves. Indeed, for a discounted-game, one can always compute an optimal strategy \cite{ZP96}, while a related question on nondeterministic discounted-sum automata, of whether the value of all words is below $0$, is not known to be decidable.   

The discounted-sum automata used in \cite{CDH10} are the same as ours, with only syntactic differences -- they use the discount-factor $\lambda$ as a multiplying factor, rather than as a dividing one, and define the value of a word as the maximal value of the automaton runs on it, rather than the minimal one. The definitions are analogous, replacing $\lambda$ with $\frac{1}{\lambda}$ and multiplying all weights by $(-1)$. In \cite{CDH10}, it is shown that for every rational discount-factor $1 < \lambda < 2$,
there is a $\lambda$-\NDA that cannot be determinized. We generalize their proof approach, in Theorem~\ref{thm:NotNaturalNumber}, extending the result to every $\lambda\in\Rat\setminus\Nat$.

The importance of approximate determinization is well known (e.g.\ \cite{AKL11,BGW01}). In \cite{AKL11}, they consider approximate determinization of sum automata with respect to ratio, showing that it is possible in cases that the nondeterministic automaton admits a ``$t$-twins'' property, which is generalization of the twins property. In \cite{BGW01}, they also consider sum automata, providing an efficient determinization algorithm, which, however, is not guaranteed to be within a certain distance from the nondeterministic automaton. In general, we are unaware of any work on the approximation of automata over \emph{infinite} words, such as discounted-sum automata. 

\section{Discounted-Sum Automata}
We consider discounted-sum automata with rational weights and rational discount factors over finite and infinite words. 

Formally, given an alphabet $\Sigma$, a {\em word\/} over $\Sigma$ is a finite or infinite sequence of letters in $\Sigma$, with $\emptyword$ for the empty word. We denote the concatenation of a finite word $u$ and a finite or infinite word $w$ by $u\con w$, or simply by $uw$.

A discounted-sum
automaton (\NDA) is a tuple $\A = \tuple{\Sigma, Q, q_{in}, \delta, \gamma, \lambda}$ over a finite alphabet $\Sigma$, with a finite set of states
$Q$, an initial state $q_{in}\in Q$, a transition function $\delta \subseteq Q \times \Sigma \times Q$, a weight function $\gamma: \delta\to
\Rat$, and a discount factor $1 < \lambda \in \Rat$. We write $\lambda$-\NDA to denote an \NDA with a discount factor $\lambda$, for example $\frac{5}{2}$-\NDA, and refer to an ``integral \NDA'' when $\lambda$ in an integer. For an automaton $\A$ and a state $q$ of $\A$, we
denote by $\A^q$ the automaton that is identical to $\A$, except for having $q$ as its initial state.

Intuitively, $\{ q' \ST (q,\sigma,q')\in\delta \}$ is the set of states that $\A$ may move
to when it is in the state $q$ and reads the letter $\sigma$.  The automaton may have many possible transitions for each state and letter, and hence we say that $\A$ is {\em nondeterministic}. In the case where for every $q \in
Q$ and $\sigma \in \Sigma$, we have that $|\{ q' \ST (q,\sigma,q')\in\delta \}| \leq 1$, we say that $\A$ is {\em deterministic}, denoted \DDA.

In the case where for every $q \in
Q$ and $\sigma \in \Sigma$, we have that $|\{ q' \ST (q,\sigma,q')\in\delta \}| \geq 1$, we say that $\A$ is {\em complete}. Intuitively, a complete automaton cannot get stuck at some state. 

In this paper, we only consider complete automata, except for Section~\ref{sec:Incomplete}, handling incomplete automata. It is natural to restrict to complete discounted-sum automata, as infinite-weight edges break the property of the decaying importance of future events.

A run of an automaton is a sequence of states and letters, $q_0, \sigma_1, q_1, \sigma_2, q_2, \ldots$, such that $q_0=q_{in}$ and for every
$i$, $(q_i,\sigma_{i+1},q_{i+1})\in\delta$. The length of a run, denoted $|r|$, is $n$ for a finite run $r = q_0, \sigma_1, q_1, \ldots,
\sigma_n, q_n$, and $\infty$ for an infinite run. 

The value of a run $r$ is $\gamma(r) = \sum_{i=0}^{|r|-1} \frac{\gamma(q_i,\sigma_{i+1},q_{i+1})}{\lambda^i}$. The value of a word $w$ (finite or infinite) is $\A(w) = \inf \{\gamma(r) \ST \mbox{$r$ is a run of $\A$ on $w$}
\}$. A run $r$ of $\A$ on a word $w$ is said to be \emph{optimal} if $\gamma(r)=\A(w)$. By the above definitions, an automaton $\A$ over finite words realizes a function from $\Sigma^*$ to $\Rat$ and over infinite words from
$\Sigma^\omega$ to $\Reals$. Two automata, $\A$ and $\A'$, are \emph{equivalent} if they realize the same function. The equivalence notion relates to either finite words or infinite words. (By Lemma~\ref{lem:FiniteEquivalenceImpliesInfiniteEquivalence}, equivalence over finite words implies equivalence over infinite words, but not vice versa.)

Next, we provide some specific definitions, to be used in the determinization construction and in the non-determinizability proofs.

The \emph{cost} of reaching a state $q$ of an automaton $\A$ over a finite word $u$ is $\Cost(q,u) = \min \{\gamma(r) \ST \mbox{$r$ is a run of $\A$ on $u$ ending in $q$} \}$, where $\min\emptyset = \infty$. The \emph{gap} of a state $q$ over a finite word $u$ is $\Gap(q,u)=\lambda^{|u|}(\Cost(q,u) - \A(u))$. Note that when $\A$ operates over infinite words, we interpret $\A(u)$, for a finite word $u$, as if $\A$ was operating over finite words.

Intuitively, the gap of a state $q$ over a word $u$ stands for the weight that a run starting in $q$ should save, compared to a run starting in $u$'s optimal ending state, in order to make $q$'s path preferable.
A gap of a state $q$ over a finite word $u$ is said to be \emph{recoverable} if there is a suffix that makes this path optimal; that is, if  there is a word $w$, such that $\Cost(q,u)+ \frac{\A^q(w)}{\lambda^{|u|}} = \A(uw)$. The suffix $w$ should be finite/infinite, depending on whether $\A$ operates over finite/infinite words. 

Notes on notation-conventions: The discount factor $\lambda$ is often used in the literature as a multiplying factor, rather than as a dividing factor, thus taking the role of $\frac{1}{\lambda}$, compared to our definitions. Another convention is to value a word as the maximal value of its possible runs, rather than the minimal value;  the two definitions are analogous, and can be interchanged by multiplying all weights by $(-1)$.

\section{Determinizability of Integral Discounted-Sum Automata}\label{sec:Determinizability}
In this section, we show that all complete \NDAs with an integral factor are determinizable. Note that the discounting factor in an \NDA is defined to be bigger than $1$. When it equals to $1$, which is the case of non-discounting sum automata, some nondeterministic automata cannot be determinized. The decision problem, of whether a non-discounting sum automaton can be determinized, is an open problem \cite{ABK11}.

Formally, we provide the following result.

\begin{thm}\label{thm:Determinizability}
For every complete $\lambda$-\NDA $\A$ with an integral factor $\lambda\in\Nat$, there is an equivalent complete $\lambda$-\DDA with up to $m^n$ states, where $m$ is the maximal difference between the weights in $\A$, multiplied by the least common denominator of all weights, and $n$ is the number of states in $\A$.
\end{thm}
\begin{proof}
Lemmas~\ref{lem:Termination}--\ref{lem:DetCorrectness}, given in the subsections below, constitute the proof.
\end{proof}

Theorem~\ref{thm:Determinizability} stands for both automata over finite words and over infinite words. 

The determinization procedure extends the subset construction, by keeping a recoverable-gap value to each element of the subset. It resembles the determinization procedure of non-discounting sum automata over finite words \cite{MohriAlgorithms,WeightedHandbook}, while having two main differences: the weight-differences between the reachable states is multiplied at every step by $\lambda$, and differences that exceed some threshold are removed.

The procedure may be used for an arbitrary $\lambda$-\NDA, always providing an equivalent $\lambda$-\DDA, if terminating. It is guaranteed to terminate for a $\lambda$-\NDA with $\lambda\in\Nat$, which is not the case for $\lambda\in\Rat\setminus\Nat$. 

%The state blow-up involved in the construction is shown to be tight for a rich alphabet of size exponential in the number of states (Theorem~\ref{thm:Tightness}). The unavoidable blow-up for an alphabet of size linear in the number of states is left as an open problem.

We start, in Subsection~\ref{sec:Construction}, with the determinization procedure, continue, in Subsection~\ref{sec:Correctness}, with its termination and correctness proofs, and conclude, in Subsection~\ref{sec:StateComplexity}, with lower bounds.

\subsection{The Construction}\label{sec:Construction}
Consider an \NDA $\A = \tuple{\Sigma, Q, q_{in}, \delta, \gamma, \lambda}$. We inductively construct an equivalent \DDA
$\D = \tuple{\Sigma, Q', q'_{in}, \delta', \gamma', \lambda}$. (An example is given in Figure~\ref{fig:Determinizing}.)

Let $T$ be the maximal difference between the weights in $\A$. That is, $T = \max \{|x-y| \ST x,y \in \range(\gamma) \}$. Since $\sum_{i=0}^\infty (\frac{1}{\lambda^i}) = \frac{1}{1-\frac{1}{\lambda}} =
\frac{\lambda}{\lambda-1} \leq 2$, we define
the set  $G = \{ v \ST v\in\Rat \mbox{ and } 0 \leq v < 2T \} \cup \{\infty \}$ of possible recoverable-gaps. The $\infty$ element denotes a non-recoverable gap, and behaves as the standard infinity element in the arithmetic operations that we will be using. Note that our discounted-sum automata do not have infinite weights; it is only used as an internal element of the construction.

A state of $\D$ extends the standard subset construction by assigning a gap to each state of $\A$. That is, for $Q=\{q_1, \ldots,
q_n \}$, a state $q'\in Q'$ is a tuple $\tuple{g_1, \ldots, g_n}$, where $g_h\in G$ for every $1 \leq h \leq n$. Intuitively, the gap $g_h$ of a state $q_h$ stands for the extra cost of reaching $q_h$, compared to the best possible value so far. This extra cost is multiplied, however, by $\lambda ^l$, for a finite run of length $l$, to reflect
the $\lambda ^l$ division in the value-computation of the suffixes. Once a gap is obviously irreducible, by being larger than or equal to $2T$, it is set to be $\infty$.

In the case that $\lambda\in\Nat$, the construction only requires finitely many elements of $G$, as shown in Lemma~\ref{lem:Termination} below, and thus it is guaranteed to terminate.

For simplicity, we assume that $q_{in}=q_1$ and extend $\gamma$ with $\gamma(\tuple{q_i,\sigma,q_j})=\infty$ for every $\tuple{q_i,\sigma,q_j}\not\in\delta$. The initial state of $\D$ is $q'_{in} = \tuple{0,\infty,\ldots,\infty}$, meaning that $q_{in}$
is the only relevant state and has a $0$ gap.

We inductively build $\D$ via the intermediate automata $\D_i= \tuple{\Sigma, Q'_i, q'_{in}, \delta'_i, \gamma'_i, \lambda}$. We start with
$\D_1$, in which $Q'_1 = \{ q'_{in} \}$, $\delta'_1 = \emptyset$ and $\gamma'_1 = \emptyset$, and proceed from $\D_i$ to $\D_{i+1}$, such that
$Q'_i \subseteq Q'_{i+1}$, $\delta'_i \subseteq \delta'_{i+1}$ and $\gamma'_i \subseteq \gamma'_{i+1}$. The construction is completed once
$\D_i = \D_{i+1}$, finalizing the desired deterministic automaton $\D = \D_i$.

In the induction step, $\D_{i+1}$ extends $\D_i$ by (possibly) adding, for every state $q'=\tuple{g_1, \ldots, g_n}\in Q'_i$ and letter $\sigma\in\Sigma$, a state $q'':= \tuple{x_1, \ldots, x_n}$, a transition $\tuple{q', \sigma, q''}$ and a weight $\gamma_{i+1}(\tuple{q', \sigma, q''}) := c$, as follows:
\begin{itemize}
\item For every $1 \leq h \leq n$, $c_h := \min \{g_j + \gamma(\tuple{q_j, \sigma, q_h}) \ST 1 \leq j \leq n  \}$
\item $c: = \min\limits_{1 \leq h \leq n}(c_h)$
\item For every $1 \leq h \leq n$, $x_h := \lambda(c_h-c)$. If $x_h \geq 2T$ then $x_h := \infty$.
\end{itemize}

\begin{figure}
% Exported from Xfig in a 46% ratio
\centering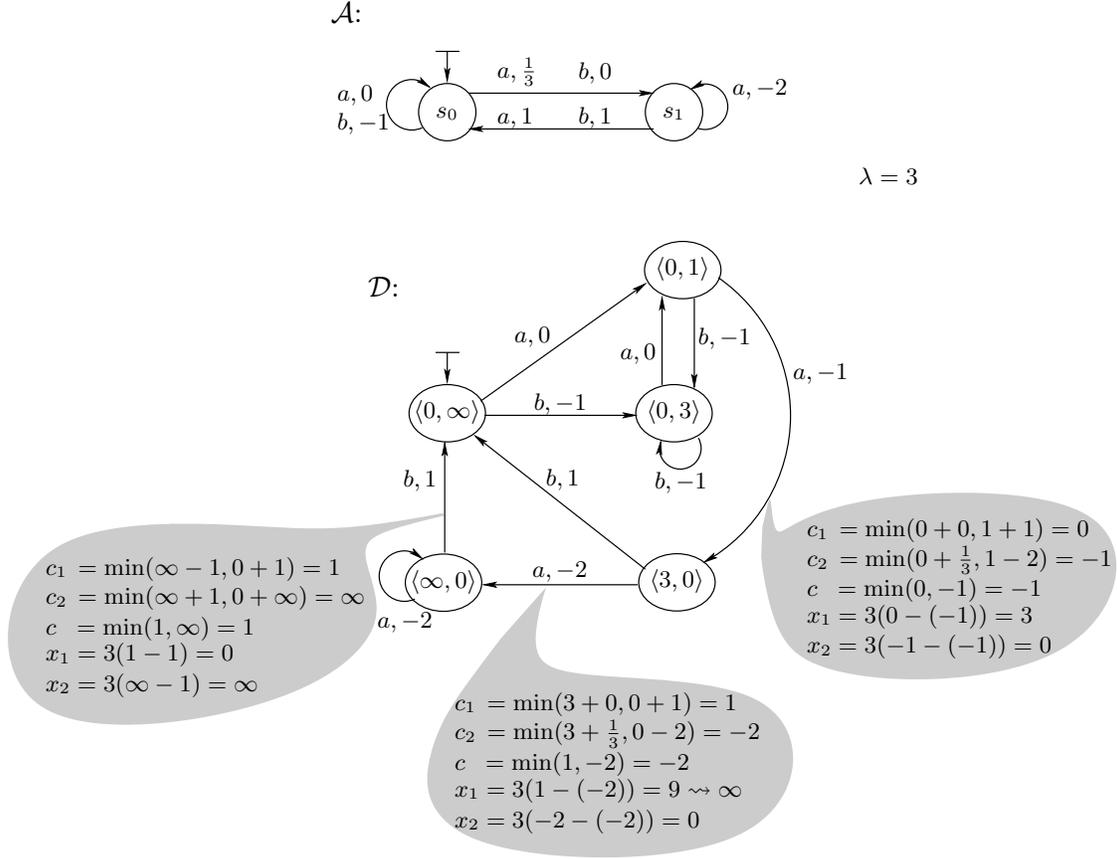 \caption{Determinizing the $3$-\NDA $\A$ into the $3$-\DDA $\D$. The gray bubbles detail some of the intermediate calculations of the determinization procedure.}\label{fig:Determinizing}
\end{figure}

\subsection{Termination and Correctness}\label{sec:Correctness}
We prove below that the above procedure always terminates for a discount factor $\lambda \in \Nat$, while generating an automaton that is equivalent to the original one. We start with the termination proof.

\begin{lem}\label{lem:Termination}
The above determinization procedure always terminates for a complete integral $\lambda$-\NDA $\A$. The resulting deterministic automaton has up to $m^n$ states, where $m$ is the maximal difference between the weights in $\A$, multiplied by the least common denominator of all weights, and $n$ is the number of states in $\A$.
\end{lem}
\begin{proof}
The induction step of the construction, extending $\D_i$ to $\D_{i+1}$, only depends on $\A$, $\Sigma$ and $Q'_i$. Furthermore, for every $i\geq 0$, we have that $Q'_i \subseteq Q'_{i+1}$. Thus, for showing the termination of the construction, it is enough to show that there is a general bound on the size of the sets $Q'_i$.  We do it by showing that the inner values, $g_1, \ldots, g_n$, of every state $q'$ of every set $Q'_i$ are from the finite set $\bar{G}$, defined below. 

Let $d\in\Nat$ be the least common denominator of the weights in $\A$, and let $m\in\Nat$ be the maximal difference between the weights, multiplied by
$d$. That is, $m= d \times \max \{|x-y| \ST x,y \in \range(\gamma) \}$. We define
the set  $\bar{G} = \{ \frac{\lambda c}{d} \ST \frac{2m}{\lambda} > c\in\Nat \} \cup \{\infty \}$ 

We start with $Q'_1$, which satisfies the property that the inner values, $g_1, \ldots, g_n$, of every state $q'\in Q'_1$ are from $\bar{G}$, as $Q'_1=\{  \tuple{0,\infty,\ldots,\infty}  \}$. We proceed by induction on the construction steps, assuming that $Q'_i$ satisfies the property. By the construction, an inner value of a state $q''$ of $Q'_{i+1}$ is derived by four operations on elements of $\bar{G}$: addition, subtraction ($x-y$, where $x \geq y$), multiplication by $\lambda\in\Nat$, and taking the minimum. 

One may verify that applying these four operations on $\infty$ and numbers of the form $\frac{\lambda c}{d}$, where $\lambda,c\in\Nat$, results in $\infty$ or in a number $\frac{v}{d}$, where $v\in\Nat$. Since the last operation in calculating an inner value of $q''$ is multiplication by $\lambda$, we have that $v$ is divisible by $\lambda$. Once an inner value exceeds $\frac{2m}{d}$, it is replaced with $\infty$. Hence, all the inner values are in $\bar{G}$. 

Having up to $m$ possible values to the elements of an $n$-tuple, provides the $m^n$ upper bound for the state space of the resulting deterministic automaton.
\end{proof}

Before proceeding to the correctness proof, we show that equivalence of automata over finite words implies their equivalence over infinite words. Note that the converse need not hold.

\begin{lem}\label{lem:FiniteEquivalenceImpliesInfiniteEquivalence}
If two \NDAs, $\A$ and $\B$, are equivalent with respect to finite words then they are also equivalent with respect to infinite words. The converse need not hold. 
\end{lem}
\begin{proof}
Assume, by contradiction, two \NDAs, $\A$ and $\B$, that are equivalent with respect to finite words and not equivalent with respect to infinite words. Then there is an infinite word $w$ and a constant number $c\neq 0$, such that $\A(w)-\B(w)=c$. Let $m$ be the maximal difference between a weight in $\A$ and a weight in $\B$. Since for every $1<\lambda$, $\sum_{i=0}^\infty (\frac{1}{\lambda^i}) = \frac{1}{1-\frac{1}{\lambda}} = \frac{\lambda}{\lambda-1} \leq 2$, it follows that the difference between the values that $\A$ and $\B$ assign to any word is smaller than or equal to $2m$. Hence, the difference between the values of their runs on suffixes of $w$, starting at a position $p$, is smaller than or equal to $\frac{2m}{\lambda^p}$. 

Now, since $\A$ and $\B$ are equivalent over finite words, it follows that they have equally-valued optimal runs over every prefix of $w$. Thus, after a long enough prefix, of length $p$ such that $\frac{2m}{\lambda^p}<c$, the difference between the values of $\A$'s and $\B$'s optimal runs on $w$ must be smaller than $c$, leading to a contradiction. 

A counter example for the converse is provided in Figure~\ref{fig:CounterExample}.
\end{proof}

\begin{figure}
% Exported from Xfig in a 50% ratio
\centering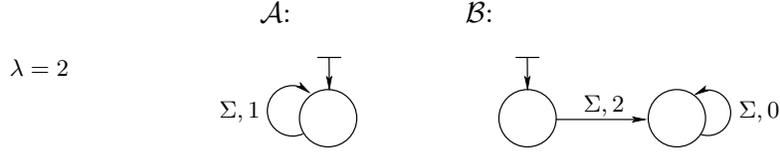 \caption{The automata $\A$ and $\B$ are equivalent with respect to infinite words, while not equivalent with respect to finite words.}\label{fig:CounterExample}
\end{figure}

We proceed with the correctness proof. By Lemma~\ref{lem:FiniteEquivalenceImpliesInfiniteEquivalence}, it is enough to prove the correctness for automata over finite words. 

Note that the correctness holds for arbitrary discount factors, not only for integral ones. For the latter, the determinization procedure is guaranteed to terminate (Lemma~\ref{lem:Termination}), which is not the case in general. Yet, in all cases that the procedure terminates, it is guaranteed to be correct.

\begin{lem}\label{lem:DetCorrectness}
Consider a $\lambda$-\NDA $\A$ over $\Sigma^*$ and a \DDA $\D$, constructed from $\A$ as above. Then, for every $w\in\Sigma^*$, $\A(w)=\D(w)$.
\end{lem}
\begin{proof}
Consider an \NDA $\A = \tuple{\Sigma, Q, q_{in}, \delta, \gamma, \lambda}$ and the \DDA 
$\D = \tuple{\Sigma, Q', q'_{in}, \delta', \gamma', \lambda}$ constructed from $\A$ as above. Let $T$ be the maximal difference between the weights in $\A$. That is, $T = \max \{|x-y| \ST x,y \in \range(\gamma) \}$. 

For a word $w$, let $q'_w=\tuple{g_1, \ldots, g_n}\in Q'$ be the last state of $\D$'s run on $w$. We show by induction on the length of the input word $w$ that:
\begin{itemize}
\item [i.]\, $\A(w) = \D(w)$.
\item [ii.] For every $1 \leq h \leq n$,  $g_h=\Gap(q_h,w)$ if $\Gap(q_h,w) < 2T$ and $\infty$ otherwise.
\end{itemize}

The assumptions obviously hold for the initial step, where $w$ is the empty word, and all values are $0$. As for the induction step, we assume they hold for $w$ and show that for every $\sigma\in\Sigma$, they hold for $w\con\sigma$. Let $q'_{w\con\sigma}=\tuple{x_1, \ldots, x_n}\in Q'$ be the last state of $\D$'s run on $w\con\sigma$. 

We start by proving the claim with respect to an \emph{infinite-state} automaton $\D'$ that is constructed as in Section~\ref{sec:Construction}, except for not changing any gap to $\infty$. Afterwards, we shall argue that changing all gaps that exceed $2T$ to $\infty$ does not harm the correctness.

\begin{enumerate}[label=\roman*.]
\item 
By the definitions of $\Cost$ and $\Gap$, we have for every $1 \leq h \leq n$, 
\begin{eqnarray*}
\Cost(q_h, w\con\sigma) &=& \min\limits_{1 \leq j \leq n} (\Cost(q_j,w) + \frac{\gamma(\tuple{q_j, \sigma, q_h}}{\lambda^{|w|}}) =  \\
&=& \min\limits_{1 \leq j \leq n} (\frac{\Gap(q_j,w)}{\lambda^{|w|}} + \A(w) + \frac{\gamma(\tuple{q_j, \sigma, q_h}}{\lambda^{|w|}}) = \\
&=&\A(w) + \frac{1}{\lambda^{|w|}} ( \min\limits_{1 \leq j \leq n} (\Gap(q_j,w) + \gamma(\tuple{q_j, \sigma, q_h}) )) = \\
&=& \mbox{By the induction assumption} =\\
&=& \D'(w) + \frac{1}{\lambda^{|w|}} ( \min\limits_{1 \leq j \leq n} (g_j + \gamma(\tuple{q_j, \sigma, q_h}) )).
\end{eqnarray*}
By the construction of $\D'$ (Section~\ref{sec:Construction}), the transition weight $c$ that is assigned on the ($|w|+1$)-step is $ c = \min\limits_{1 \leq h \leq n}  ( \min\limits_{1 \leq j \leq n} (g_j + \gamma(\tuple{q_j, \sigma, q_h}) )) $. Therefore, 
\begin{eqnarray*}
\D'(w\con\sigma) &=& \D'(w) + \frac{c}{\lambda^{|w|}} =\\
&=&  \D'(w) + \frac{1}{\lambda^{|w|}} \min\limits_{1 \leq h \leq n}   \min\limits_{1 \leq j \leq n} (g_j + \gamma(\tuple{q_j, \sigma, q_h}) )= \\
&=&  \min\limits_{1 \leq h \leq n} (\D'(w) +   \frac{1}{\lambda^{|w|}}  \min\limits_{1 \leq j \leq n} (g_j + \gamma(\tuple{q_j, \sigma, q_h}) ))=\\
&=&  \min\limits_{1 \leq h \leq n} \Cost(q_h, w\con\sigma) =\\
&=&  \A(w\con\sigma).
\end{eqnarray*}

\item   We use the notations and the equations of part (i.) above. By the construction of $\D'$, for every $1 \leq h \leq n$, 
\begin{eqnarray*}
x_h &=& \lambda (\min\limits_{1 \leq j \leq n} (g_j + \gamma(\tuple{q_j, \sigma, q_h}) ) - c)=\\
&=& \lambda (\min\limits_{1 \leq j \leq n} (g_j + \gamma(\tuple{q_j, \sigma, q_h}) ) - \lambda^{|w|}(\D'(w\con\sigma) - \D'(w)))=\\
&=& \lambda (\lambda^{|w|}(\Cost(q_h, w\con\sigma) - \D'(w) ) - \lambda^{|w|}(\D'(w\con\sigma) - \D'(w)))=\\
&=& \lambda^{|w|+1} (\Cost(q_h, w\con\sigma) - \D'(w\con\sigma) )=\\
&=& \lambda^{|w|+1} (\Cost(q_h, w\con\sigma) - \A(w\con\sigma) )=\\
&=& \Gap(q_h,w\con\sigma).
\end{eqnarray*}

\end{enumerate}

\noindent It is left to show that the induction is also correct for the \emph{finite-state} automaton $\D$. The only difference between the construction of $\D$ and of $\D'$ is that the former changes all gaps $(g_j$) above $2T$ to $\infty$. We should thus show that if $g_j$, for some $1 \leq j \leq n$, exceeds $2T$ at a step $i$ of the construction, and this $g_j$ influences $g_h$, for some $1 \leq h \leq n$, at step $i+1$, then $g_h \geq 2T$. This implies that $\D(w)=D'(w)$, since at every step of the construction there is at least one $1\leq h \leq n$, such that $g_h=0$, corresponding to an optimal run of $\A$ on $w$ ending in state $q+h$. Formally, we should show that if $g_h = \lambda (g_j + \gamma(\tuple{q_j, \sigma, q_h}) - c)$, where $c$ is the transition weight assigned in the construction on the $i+1$ step (as defined in part (i.) above), then $g_h \geq 2T$. Indeed, $g_h \geq \lambda (2T + \gamma(\tuple{q_j, \sigma, q_h}) - c) \geq 2 (2T + \gamma(\tuple{q_j, \sigma, q_h}) - c) \geq 2 (2T + (-T) ) = 2T$.
\end{proof}

\subsection{State Complexity}\label{sec:StateComplexity}
For an integral \NDA $\A$, the deterministic automaton constructed as in Subsection~\ref{sec:Construction} has up to $m^n$ states, where $m$ is the maximal difference between the weights in $\A$, multiplied by the least common denominator of all weights, and $n$ is the number of states in $\A$ (Lemma~\ref{lem:Termination}). We provide below corresponding lower bounds. The lower bounds are with respect to automata that operate over infinite words, and by Lemma~\ref{lem:FiniteEquivalenceImpliesInfiniteEquivalence} they also apply to automata operating over finite words.

\Paragraph{Dependency on the number of states}
Unavoidable exponential dependency of the determinization on the number of states ($n$) is straightforward, by considering discounted-sum automata as generalizing finite automata over finite words. We demonstrate this generalization in Figure~\ref{fig:StateLowerBound}, showing how to translate the finite automaton $\A_k$ over the alphabet $\{a,b\}$ into an \NDA $\A'_k$ over the alphabet $\{a,b,\#\}$ with transition weights in $\{0,1\}$. By this translation, the value that $\A'_k$ assigns to an infinite word $w$ is smaller than $0$ if and only if the $\#$ sign appears in $w$ and the prefix of $\A'_k$ up to the first $\#$ sign is accepted by $\A_k$. The automaton $\A_k$ has $k+1$ states and accepts the language $L_k$ of finite words that have an `$a$' in the last-by-$k$ position. (For example, $L_1$ accepts the words with an `$a$' as the penultimate letter.) It is known that a deterministic automaton $\D_k$ for $L_k$ must have at least $2^k$ states. Assuming, by contradiction, a \DDA $\D'_k$ that is equivalent to $\A'_k$ and has less than $2^k$ states, will easily allow to construct a deterministic automaton $\D_k$ for $L_k$ with less than $2^k$ states: By the structure of $\A'_k$, it follows that $\D_k$ must have a $0$ weight in all transitions that occur before a $\#$ sign, and a weight of $-1$ in some of the transitions upon reading $\#$. Hence, one can translate $\D'_k$ into the required finite automaton $\D_k$ by setting the accepting states to be the states that have an outgoing transition with a weight of $-1$.

\begin{figure}
% Exported from Xfig in a 50% ratio
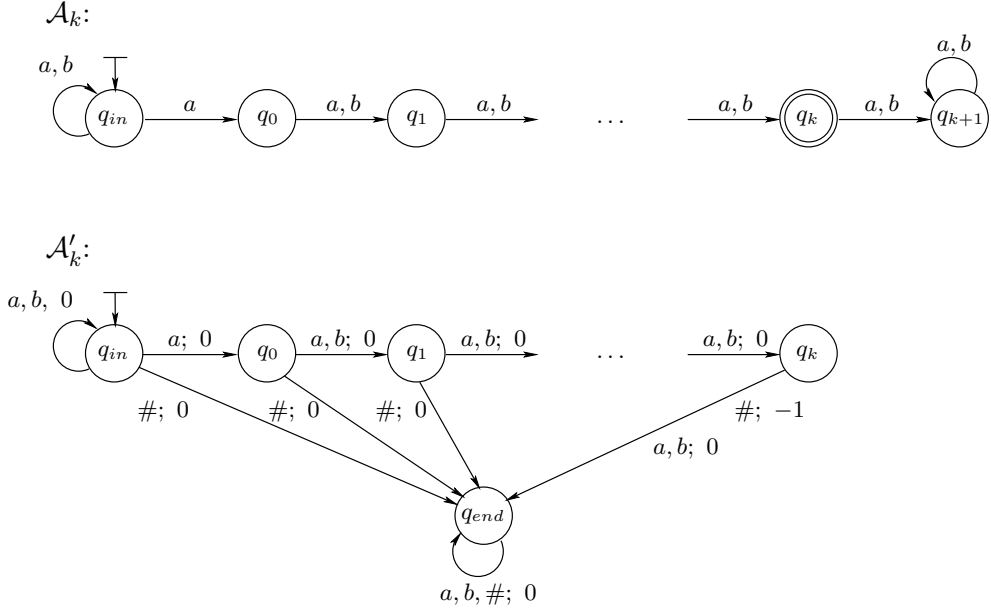 \caption{The family $\A_k$ of finite automata accepts finite words that have an `$a$' in the last-by-$k$ position. The family $\A'_k$ of \NDAs is their direct generalization to discounted-sum automata over infinite words, where the $\#$ sign marks the ``end'' of the word. This generalization allows to show the exponential dependency on the number of states in determinizing discounted-sum automata.}\label{fig:StateLowerBound}
\end{figure}

\begin{prop}
For every integral discount factor $\lambda$ there is a $\lambda$-\NDA with $k+3$ states over the alphabet $\{a,b,\#\}$ and with weights in $\{0,1\}$, such that every equivalent $\lambda$-\DDA must have at least $2^k$ states.
\end{prop}

\Paragraph{Dependency on the weights}
The dependency of the determinization construction on the weight difference ($m$) is linear. Yet, a reasonable description of the automaton weights is not unary, but, for example, binary, making the construction exponentially dependent on the weight description. We show below that the determinization must indeed depend on the weight value, making it possibly exponential in the weight description. For every fixed discount factor $\lambda$, we provide a family of automata $\A_k$, for $k\in \{ \lambda, \lambda+1, \lambda+2,\ldots\}$ (Figure~\ref{fig:WeightLowerBound}) over a fixed alphabet and weights in $(-k,1)$, such that $\A_k$ has three states and an equivalent deterministic automaton must have at least $k-\lambda$ states (no matter how concise the description of $k$ in $\A_k$ is).

We start by providing a sufficient condition, under which two different gaps must be associated with two different states of a deterministic automaton. The lemma below generalizes an argument given in \cite{CDH10}.

\begin{lem}\label{lem:DifferentGaps}
Consider an \NDA $\A$ for which there is an equivalent \DDA $\D$. If there is a state $q$ of $\A$, finite words $u$ and $u'$, and words $w$ and $z$, such that:
\begin{enumerate}[label=\roman*.]
\item $\A$ has runs on $u$ and on $u'$ ending in $q$; 
\item $\Gap(q,u)\neq\Gap(q,u')$; 
\item The gaps of $q$ over both $u$ and $u'$ are recoverable with $w$, that is, $\A(uw) = \Cost(q,u)+\frac{\A^q(w)}{\lambda^{|u|}}$ and  $\A(u'w) = \Cost(q,u')+\frac{\A^q(w)}{\lambda^{|u'|}}$; and
\item $\A$ is ``indifferent'' to concatenating $z$ to $u$ and to $u'$, that is $\A(uz) = \A(u)$ and $\A(u'z) = \A(u')$
 \end{enumerate}
then the runs of $\D$ on $u$ and on $u'$ end in different states.

The words $w$ and $z$ should be finite for automata over finite words and infinite for automata over infinite words. In the former case, $z$ is redundant as it can always be $\emptyword$.
\end{lem}
\begin{proof}
Consider the above setting. Then, we have that $\A(uw) - \A(uz) = \A(uw) - \A(u) = \Cost(q,u)+\frac{\A^q(w)}{\lambda^{|u|}} -  \A(u) = ( \Cost(q,u) -  \A(u) ) + \frac{\A^q(w)}{\lambda^{|u|}} = 
\frac{\Gap(q,u) + \A^q(w)}{\lambda^{|u|}}$ and analogously $\A(u'w) - \A(u'z) = \frac{\Gap(q,u') + \A^q(w)}{\lambda^{|u'|}}$. Thus, 
$$(I)~~\Gap(q,u) = \lambda^{|u|}[\A(uw) - \A(uz)]-\A^q(w); ~~~ \Gap(q,u') = \lambda^{|u'|}[\A(u'w) - \A(u'z)]-\A^q(w)$$

Now, assume, by contradiction, a single state $p$ of $\D$ in which the runs of $\D$ on both $u$ and $u'$ end. Then, we have that 
$$(II)~~\D(uw) - \D(uz) = \frac{\D^p(w)}{\lambda^{|u|}};~~~ \D(u'w) - \D(u'z) = \frac{\D^p(w)}{\lambda^{|u'|}}$$
Since $\A$ and $\D$ are equivalent, we may replace between $[\A(uw) - \A(uz)]$ and $[\D(uw) - \D(uz)]$ as well as between $[\A(u'w) - \A(u'z)]$ and $[\D(u'w) - \D(u'z)]$. Making the replacements in equations (I) above, we get:
$$(I\&II)~~\Gap(q,u) = \lambda^{|u|}\frac{\D^p(w)}{\lambda^{|u|}}-\A^q(w); ~~~ \Gap(q,u') = \lambda^{|u'|}\frac{\D^p(w)}{\lambda^{|u'|}}-\A^q(w)$$
Therefore, $\Gap(q,u)=\Gap(q,u')$, leading to a contradiction.
\end{proof}

We continue with the lower bound with respect to the weight difference, showing that the blow-up in the determinization depends on the weight value and not on the weight description.

\begin{figure}
% Exported from Xfig in a 50% ratio
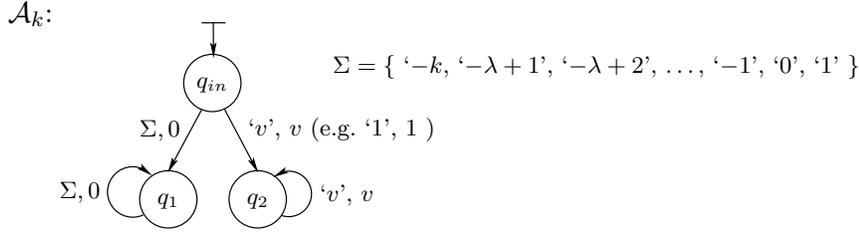 \caption{The family $\A_k$ of \NDAs with weights in $(-k,1)$, where for every $k$, a deterministic automaton equivalent to $\A_k$ must have at least $k-\lambda$ states.}\label{fig:WeightLowerBound}
\end{figure}

\begin{thm}\label{thm:WeightLowerBound}
For every integral discount factor $\lambda \geq 2$ and number $k>\lambda$ there is a $\lambda$-\NDA with three states, weights in $\{-k, -\lambda+1, -\lambda+2, \ldots, -1, 0, 1\}$ over an alphabet of size $\lambda +2$, such that every equivalent $\lambda$-\DDA must have at least $k-\lambda$ states.
\end{thm}
\begin{proof}
For every $\lambda \geq 2$ and $k>\lambda$, we define the $\lambda$-\NDA $\A_k= \tuple{\Sigma, Q, q_{in}, \delta, \gamma, \lambda}$, as illustrated in Figure~\ref{fig:WeightLowerBound}, where
\begin{itemize}
\item $\Sigma= \{-k, -\lambda+1, -\lambda+2, \ldots, -1, 0, 1\}$
\item $Q=\{ q_{in}, q_1, q_2\}$
\item $\delta = \{ \tuple{q_{in},\sigma,q_1}, \tuple{q_{in},\sigma,q_2},\tuple{q_1,\sigma,q_1},\tuple{q_2,\sigma,q_2} \ST \sigma\in\Sigma  \}$
\item For every $\sigma\in\Sigma$ and $q\in Q$: $\gamma(\tuple{q,\sigma,q_1})=0$ and $\gamma(\tuple{q,\sigma,q_2})=\sigma$
\end{itemize}
Note that, for simplicity, we define the alphabet letters of $\Sigma$ as numbers, denoting the letter of a number $n$ by `$n$'. 

For every integer $\lambda < x \leq k$, we show, by induction on $x$, that there is a finite word $u_x$, such that $\Gap(q_2,u_x)=x$. Intuitively, $u_x$ is the representation of $x$ in base $\lambda$. Formally, for the base case, we have $x=\lambda$ and $u_\lambda=\letter{1}$. For the induction step, let $y=\uwhole{\frac{x}{\lambda}}$. Then, $u_x=u_y \con \letter{x - \lambda y}$.

Now, for every $\lambda < i < j \leq k$, we have that $u_i$ and $u_j$ satisfy the conditions of Lemma~\ref{lem:DifferentGaps}, by having $u=u_i$, $u'=u_j$, $z=\letter{0}^\omega$, and $w=\letter{-k}^\omega$. Hence, a \DDA equivalent to $\A_k$ has two different states corresponding to each two different integers in $\{\lambda,\ldots, k \}$, and we are done.
\end{proof}

\Paragraph{Dependency on the combination of states and weights}
The exponential dependency on the number of states ($n$) and on the weight description ($\log m$), as discussed above, provides a lower bound of $2^{\max(n,~\log m)}$. For showing that the construction depends on $2^{n \log m} = m^n$, we use a rich alphabet of size in $O(m^n)$. For an alphabet of size linear in $m$ and $n$, the exact unavoidable state blow-up is left as an open problem. 
A family of automata $\A_{k,l}$, with which we provide this lower bound, is illustrated in Figure~\ref{fig:LowerBound}. 
%Intuitively, the rich alphabet allows to set every gap in $\{0, 1, 2, \ldots, m+1 \}$ to each of the $n$ states. Two different gaps have, under the appropriate setting, two suffixes that distinguish between them. Hence, an equivalent deterministic automaton must have a unique state for each recoverable-gap, yielding at least $m^n$ states.

\begin{figure}
% Exported from Xfig in a 50% ratio
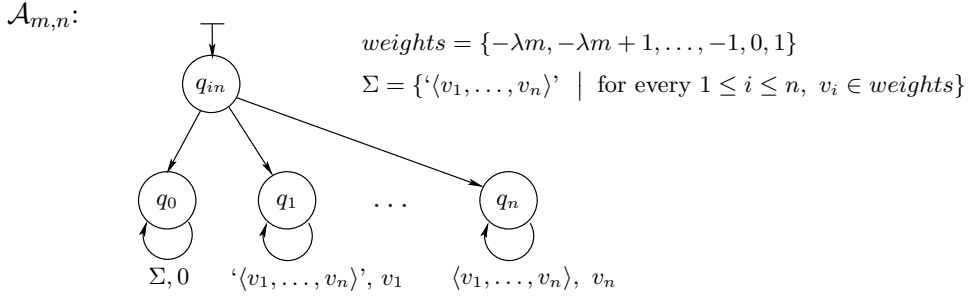 \caption{The family of integral \NDAs, where for every $k$ and $l$, a deterministic automaton equivalent to $\A_{k,l}$ must have at least $k^l$ states.}\label{fig:LowerBound}
\end{figure}

\begin{thm}\label{thm:Tightness}
For every $\lambda, k, l\in\Nat$, there is a $\lambda$-\NDA with $l+2$ states and weights in $\{ -\lambda k, -\lambda k+1,\ldots, -1, 0, 1\}$, such that every equivalent  \DDA has at least $k^l$ states.
\end{thm}
\begin{proof}
For every $\lambda, k, l\in\Nat$, we define the \NDA $\A_{k,l} = \tuple{\Sigma, Q, q_{in}, \delta, \gamma, \lambda}$, as illustrated in Figure~\ref{fig:LowerBound}, where:
\begin{itemize}
\item $\Sigma= \{ \tuple{v_1,\ldots v_l} \ST \mbox{ for every } 1\leq i \leq l, v_i \in \{ -\lambda k, -\lambda k+1,\ldots, -1, 0, 1\} \}$
\item $Q=\{ q_{in}, q_0, q_1, \ldots, q_l\}$
\item $\delta = \{ \tuple{q_{in},\sigma,q_i}, \tuple{q_i,\sigma,q_i}  \ST 0 \leq i \leq l \mbox{ and } \sigma\in\Sigma  \}$
\item For every $\sigma=\tuple{v_1,\ldots v_l}\in\Sigma$ and $1 \leq i \leq l$: $\gamma(\tuple{q_{in},\sigma,q_0})=0$, $\gamma(\tuple{q_{in},\sigma,q_i})=0$, $\gamma(\tuple{q_0,\sigma,q_0})=0$ and $\gamma(\tuple{q_i,\sigma,q_i})=v_i$
\end{itemize}
Note that, for simplicity, we define the alphabet letters of $\Sigma$ as tuples of numbers. 

Consider a \DDA $\D$ equivalent to $A_{k,l}$. We will show that there is a surjective mapping between $\D$'s states and the set of vectors $V =  \{ \tuple{g_1,\ldots,g_l} \ST \mbox{ for every } 1 \leq i \leq l, 1 \leq g_i \leq k \}$. 

We call an $l$-vector of gaps, $G = \tuple{g_1,\ldots,g_l}$, a \emph{combined-gap}, specifying the gaps of $q_1,\ldots,q_l$, respectively. Due to the rich alphabet, for every combined-gap $G \in V$, there is a finite word $u_G$, such that for every $1 \leq i \leq l$, $\Gap(q_i,u_G)=g_i$. 

Every two different combined gaps, $G$ and $G'$, are different in at least one dimension $j$ of their $l$-vectors. Thus, $A_{k,l}$ satisfies the conditions of Lemma~\ref{lem:DifferentGaps}, by having $u=u_G$, $u'=u_{G'}$, $z=$`$\tuple{0,\ldots, 0}$'$^\omega$, and $w=$`$\tuple{0,\ldots 0, -\lambda k, 0, \ldots 0}$'$^\omega$, where the repeated letter in $w$ has $0$ in all dimensions except for the $j$'s dimension, in which it has $-\lambda k$. Hence, $A_{k,l}$ has two different states corresponding to each two different vectors in $V$, and we are done.
\end{proof}

\section{Nondeterminizability of Nonintegral Discounted-Sum Automata}\label{sec:NonDeterminizability}
The discount-factor $\lambda$ plays a key role in the question of whether a complete $\lambda$-\NDA is determinizable. In Section~\ref{sec:Determinizability}, we have shown that an integral factor guarantees the automaton's determinizabilty. In Subsection~\ref{sec:Complete} below, we show the converse for every nonintegral factor. 

In the whole paper, except for Subsection~\ref{sec:Incomplete} below, we only consider complete automata. In Subsection~\ref{sec:Incomplete}, we show that once allowing incomplete automata or, equivalently, adding infinite weights, there is a non-determinizable automaton for every discount-factor $\lambda$, including integral ones. 

\subsection{Complete Automata}\label{sec:Complete}
We show below that for every noninntegral discount factor $\lambda$, there is a complete $\lambda$-\NDA that cannot be determinized. The proof generalizes the approach taken in \cite{CDH10}, where the case of $1<\lambda<2$ was handled. It is shown for automata over infinite words, and by Lemma~\ref{lem:FiniteEquivalenceImpliesInfiniteEquivalence} it also applies to automata over finite words.

Intuitively, for a discount factor that is not a whole number, a nondeterministic automaton might have infinitely many recoverable-gaps, arbitrarily close to each other. Two different gaps have, under the appropriate setting, two suffixes that distinguish between them (Lemma~\ref{lem:DifferentGaps}). Hence, an equivalent deterministic automaton must have a unique state for each recoverable-gap, which is impossible for infinitely many gaps.

\begin{figure}
% Exported from Xfig in a 50% ratio
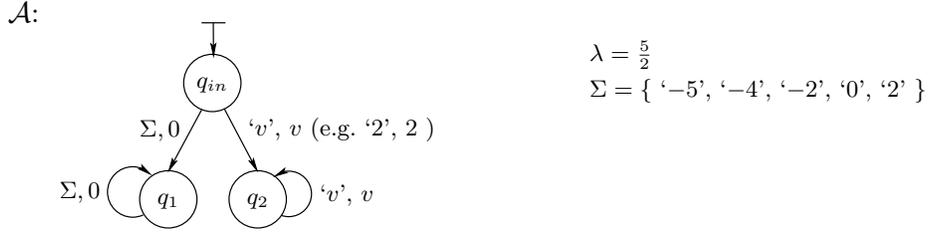 \caption{The non-determinizable $\frac{5}{2}$-\NDA $\A$.}\label{fig:NonDeterminizable}
\end{figure}

\begin{thm}\label{thm:NotNaturalNumber}
For every nonintegral discount factor $1<\lambda\in\Rat\setminus\Nat$, there is a complete $\lambda$-\NDA for which there is no equivalent \DDA (with any discount factor).
\end{thm}
\begin{proof}
For every $1<\lambda\in\Rat\setminus\Nat$, we define a complete $\lambda$-\NDA $\A = \tuple{\Sigma, Q, q_{in}, \delta, \gamma, \lambda}$ and show that $\A$ is not determinizable.
The automaton $\A$ operates over infinite words, and by Lemma~\ref{lem:FiniteEquivalenceImpliesInfiniteEquivalence} it also applies to automata operating over finite words.

Let $\lambda=\frac{h}{k}$, where $h$ and $k$ are mutually prime, and define:
\begin{itemize}
\item $\Sigma= \{ -jk  \ST  j\in\Nat \mbox{ and } jk < h\} \cup \{-h, k\}$
\item $Q=\{ q_{in}, q_1, q_2\}$
\item $\delta = \{ \tuple{q_{in},\sigma,q_1}, \tuple{q_{in},\sigma,q_2},\tuple{q_1,\sigma,q_1},\tuple{q_2,\sigma,q_2} \ST \sigma\in\Sigma  \}$
\item For every $\sigma\in\Sigma$ and $q\in Q$: $\gamma(\tuple{q,\sigma,q_1})=0$ and $\gamma(\tuple{q,\sigma,q_2})=\sigma$
\end{itemize}
Note that, for simplicity, we define the alphabet letters of $\Sigma$ as numbers, denoting the letter of a number $n$ by `$n$'. The \NDA $\A$ for $\lambda=\frac{5}{2}$ is illustrated in Figure~\ref{fig:NonDeterminizable}.

We show that $\A$ cannot be determinized by providing an infinite word $w$, such that $q_2$ has a unique recoverable gap for each of $w$'s prefixes. By Lemma~\ref{lem:DifferentGaps}, such a word $w$ implies that $\A$ cannot be determinized, as each of its prefixes can be continued with either `$0$'$^\omega$ or with a suffix that recovers $q_2$'s gap.

We inductively define $w$, denoting its prefix of length $i$ by $w_i$, as follows: the first letter is `$k$' and the $i+1$'s letter is `$-jk$', such that $0 \leq \Gap(q_2,w_i)\frac{h}{k} - jk \leq k$. Intuitively, each letter is chosen to almost compensate on the gap generated so far, by having the same value as the gap up to a difference of $k$. 

We show that $w$ has the required property, by proving the following three claims, each being a step toward proving the next claim.
\begin{enumerate}
\item The word $w$ is infinite and $q_2$ has a recoverable-gap for each of its prefixes.
\item There is no prefix of $w$ for which $q_2$'s gap is $0$.
\item There are no two different prefixes of $w$ for which $q_2$ has the same gap.
\end{enumerate}
Indeed:
\begin{enumerate}
\item Since $\gamma(\tuple{q_2, -h, q_2})=-h$, a gap $g$ of $q_2$ is obviously recoverable if $g\leq h$. We show by induction on the length of $w$'s prefixes that for every $i\geq 1$, we have that $\Gap(q_2, w_i)\leq h$. It obviously holds for the initial step, as $w_1=$`$k$' and $\Gap(q_2, w_1) = k\frac{h}{k}=h$. Assuming that it holds for the $i$'s prefix, we can choose the $i+1$'s letter to be some `$-jk$' $ \in \Sigma$, such that $0 \leq \Gap(q_2, w_i) - jk \leq k$. Hence, we get that 
\begin{equation}\label{eq:ChosenLetter}
\Gap(q_2, w_{i+1}) =  (\Gap(q_2, w_i) - jk)\frac{h}{k} \leq h.
\end{equation}

\item Assume, by contradiction, a prefix of $w$ of length $n$ whose recoverable-gap is $0$.
We have then, by Equation~\ref{eq:ChosenLetter}, that:
$$((( h -j_1 k)\frac{h}{k} -j_2 k)\frac{h}{k} \ldots -j_{n-1} k)\frac{h}{k} = 0$$
for some $j_1,\ldots,j_n\in\Nat$. Simplifying the equation, we get that 
$$ \frac{h^n - j_1 k h^{n-1} - j_2 k^2 h^{n-2} - \ldots - j_{n-1} k^n}{k^{n-1}} = 0$$
Therefore, $h^n =  j_1 k h^{n-1} + \ldots + j_{n-1} k^n$.
Now, since $k$ divides $j_1 k h^{n-1} + \ldots + j_{n-1} k^n$, it follows that $k$ divides $h^n$, which leads to a contradiction, as $h$ and $k$ are mutually prime.

\item Assume, by contradiction, that $q_2$ has the same gap $x$ for two prefixes, $n\geq 1$ steps apart. We have then, by Equation~\ref{eq:ChosenLetter}, that:
$$((((x-j_1 k)\frac{h}{k} -j_2 k)\frac{h}{k} -j_3 k)\frac{h}{k} \ldots -j_n k)\frac{h}{k} = x$$
for some $j_1,\ldots,j_n\in\Nat$. Simplifying the equation, we get that 
$$ \frac{x h^n - j_1 k h^n - j_2 k^2 h^{n-1} - \ldots - j_n k^nh}{k^n} = x$$
Thus, 
$$ x( h^n - k^n) = j_1  k h^n + j_2 k^2 h^{n-1} + \ldots + j_n k^nh$$
Hence, $x (h^n - k^n)$ is an integer, and since $x\neq 0$ and the right side of the above equation is divisible by $k$, so is $x (h^n - k^n)$.

Let us take a closer look at the gap $x$, assuming that it is generated, in its first occurrence, by a prefix of $w$ of some length $m$. Following Equation~\ref{eq:ChosenLetter}, $x=\frac{a}{k^{m-1}}$ for some integer $a$. We claim that $a$ is co-prime with $k$, and show it by induction on the length of $w$'s prefix with which the gap is associated. For the base case, the gap is $\frac{h}{k^0}$, and the claim holds, as $h$ and $k$ are co-prime. Assume that the claim holds for a prefix of length $m-1$ with a gap $\frac{a'}{k^{m-2}}$. Then, the next gap, following  Equation~\ref{eq:ChosenLetter}, is $\frac{a}{k^{m-1}} = (\frac{a'}{k^{m-2}} - jk)\frac{h}{k} = \frac{h(a' - j k^{m-1})}{k^{m-1}}$, implying that $a = h(a'-jk^{m-1})$. Since $h$ and $a'$ are co-prime with $k$, while $jk^{m-1}$ is divisible by $k$, it follows that $a$ is co-prime with $k$, and the induction proof is done.

Now, we have by the above that $k$ divides $x (h^n - k^n) = \frac{a}{k^{m-1}}(h^n - k^n)$, while $a$ is co-prime with $k$. Therefore, by Euclid's lemma, $k$ divides $h^n - k^n$. But, since $k$ divides $k^n$, it follows that $k$ also divides $h^n$, which leads to a contradiction.\qedhere
\end{enumerate}
\end{proof}

\subsection{Incomplete Automata}\label{sec:Incomplete}
Once considering incomplete automata or, equivalently, automata with $\infty$-weights, or automata where some of the states are accepting and some are not, no discount factor can guarantee determinization. The reason is that there is no threshold above which a gap becomes irrecoverable -- no matter how (finitely) bad some path is, it might eventually be essential, in the case that the other paths get stuck.  

Formally:
\begin{thm}
For every rational discount factor $\lambda$, there is an incomplete $\lambda$-\NDA for which there is no equivalent \DDA (with any discount factor).
\end{thm}
\begin{proof}
Consider the incomplete automaton $\B$ presented in Figure~\ref{fig:Incomplete} with a discount factor $\lambda\in\Rat$. 

For every $n\in\Nat$, we have that $\Gap(q_2,a^n)=\sum_{i=0}^n \lambda^i$. Since $q_1$ has no transition for the letter $b$, it follows that all these gaps are recoverable. Hence, for every $i,j\in\Nat$ such that $i\neq j$, we satisfy the conditions of Lemma~\ref{lem:DifferentGaps} with $u=a^i$, $u'=a^j$, $z=a^\omega$ and $w=b^\omega$ (for automata over finite words, $z=\emptyword$ and $w=b$). Therefore, an equivalent deterministic automaton must have infinitely many states, precluding its existence.
\end{proof}

\begin{figure}
% Exported from Xfig in a 50% ratio
\centering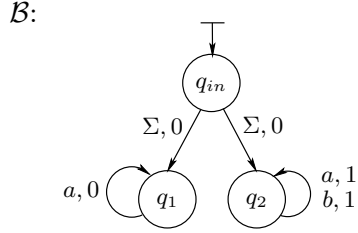 \caption{The incomplete automaton $\B$ is not determinizable with respect to any discount-factor.}\label{fig:Incomplete}
\end{figure}

\section{Approximate Determinization}\label{sec:ApproxDet}

As shown in Section~\ref{sec:NonDeterminizability}, nonintegral \NDAs cannot, in general, be determinized. Yet, by their discounting behavior, they can always be determinized approximately. That is, for every \NDA, there is a \DDA, such that the difference between their values, on all words, is as small as required. The naive construction of the deterministic automaton is achieved by unfolding the computations of the nondeterministic automaton up to a sufficient level. The size of the constructed automaton depends on the required precision and on the proximity of the discount-factor is to $1$. We represent the precision by $\varepsilon=2^{-p}$ and the discount factor by $\lambda=1+2^{-k}$, for positive integers $p$ and $k$. We analyze the unfolding construction to generate an automaton whose state space is exponential in $p$ and doubly exponential in $k$. We then provide an alternative construction, by generalizing the determinization procedure of Section~\ref{sec:Construction}. We show that our construction is singly exponential in $k$, in $p$, and in the number of states of the automaton. We conclude the section by proving matching lower bounds, showing exponential dependency on each of these three parameters.

We start with an interesting observation on discounting and half life time: for every integer $K \geq 2$, the half life time with respect to the discount factor $1+\frac{1}{K}$, meaning the number of time units before the discounting gets to $2$, is roughly $K$. More precisely, as $K$ tends to infinity, $(1+\frac{1}{K})^K$ is exactly $e \,(\approx 2.72)$. Note that we can take advantage of this property, as we represent the discount factor by $1+2^{-k}$, which equals to $1+\frac{1}{K}$, for $K = 2^k$. For our purposes, we show in Lemma~\ref{lem:HalfTime} below that $(1+\frac{1}{K})^K$ is always between $2$ and $3$, as well as a corresponding bound for $\log (1+\frac{1}{K})$.

\begin{lem}\label{lem:HalfTime}
For every integer $K\geq2$, we have:
\begin{enumerate}
\item $1 < K \log (1+ \frac{1}{K}) < \frac{3}{2}$.
\item $2 < (1+\frac{1}{K})^K < 3$.
\end{enumerate}
\end{lem}
\begin{proof}\
\begin{enumerate}
\item We use the Mercator series, which is the Taylor series for the natural logarithm, stating that for every $-1 < x \leq 1$, $\ln (1+x) = x - \frac{x^2}{2} + \frac{x^3}{3} - \frac{x^4}{4} + \ldots$. Setting $x=\frac{1}{K}$, we get that $\ln (1+\frac{1}{K}) = \frac{1}{K} - \frac{1}{2K^2} + \frac{1}{3K^3} - \frac{1}{4K^4} + \ldots$. Thus, $$K \ln (1+ \frac{1}{K}) =1 - \frac{1}{2K} + \frac{1}{3K^2} - \frac{1}{4K^3} + \frac{1}{5K^4} - \frac{1}{6K^5} + \ldots.$$

Since for every positive integer $i$, $(- \frac{1}{iK^{(i-1)}} + \frac{1}{(i+1)K^i}) < 0$, it follows that the above series is smaller than $1$. Analogously, since for every positive integer $i$, $\frac{1}{iK^{(i-1)}} - \frac{1}{(i+1)K^i} > 0$, it follows that the above series is bigger than $1- \frac{1}{2K}$. Hence, $1- \frac{1}{2K} < K \ln (1+ \frac{1}{K}) < 1 $. Therefore, as $K$ tends to infinity, $K \log (1+ \frac{1}{K})$ converges to $\log e$, where $e$ is Euler's constant. Specifically, for every $K\geq2$, we have $1 < K \log (1+ \frac{1}{K}) < \frac{3}{2}$.

\item Let $z=\log (1+ \frac{1}{K})$.  we have $(1+\frac{1}{K})^K = (1+\frac{1}{K})^\frac{zK}{z} = ((1+\frac{1}{K})^\frac{1}{z})^{zK} = 2^{zK}$.

From the first part of the lemma, we know that $\frac{1}{K} < z < \frac{3}{2K}$. Thus, $2 = 2^\frac{K}{K} < 2^{zK} < 2^\frac{3K}{2K} < 3$. Hence, $2 <  (1+\frac{1}{K})^K < 3$.\qedhere
\end{enumerate}
\end{proof}

\subsection{Approximate Automata}
We define that an automaton can be determinized approximately if for every real precision $\varepsilon>0$, there is a deterministic automaton such that the difference between their values on all words is less than or equal to $\varepsilon$. 
Formally,
\begin{defi}[Approximation]\label{def:Approx}\
\begin{itemize}
\item An \NDA $\A'$ \emph{$\varepsilon$-approximates} an \NDA $\A$, for a real constant $\varepsilon > 0$, if for every word $w$, $|\A(w) - \A'(w)| \leq \varepsilon$. 
\item  An \NDA $\A$ can be \emph{determinized approximately} if for every real constant $\varepsilon >0$ there is a \DDA  $\A' $ that $\varepsilon$-approximates $\A$.
\end{itemize}
\end{defi}

The relation between \NDAs on finite words and on infinite words, as stated in Lemma~\ref{lem:FiniteEquivalenceImpliesInfiniteEquivalence}, follows to approximated automata, meaning that approximation over finite words guarantees approximation over infinite words, but not vice versa. Intuitively, as the influence of word suffixes is decaying, the distance between two automata cannot change ``too much'' after long enough prefixes. Hence, if the automata are close enough for every finite prefix, so they are for an entire infinite word. As for the converse, the distance between the automata might gradually decrease, only converging at the infinity.

\begin{lem}\label{lem:FiniteApproxImpliesInfiniteApprox}
For every precision $\varepsilon>0$ and discount factor $\lambda>1$, if a $\lambda$-\NDA $\varepsilon$-approximates another $\lambda$-\NDA over finite words then it also $\varepsilon$-approximates it over infinite words. The converse need not hold. 
\end{lem}
\begin{proof}
Assume, by contradiction, a precision $\varepsilon>0$, a discount factor $\lambda>1$, and two $\lambda$-\NDAs, $\A$ and $\B$, such that $\B$ $\varepsilon$-approximates $\A$ with respect to finite words but not with respect to infinite words.

Then there is an infinite word $w$, such that $|\A(w)-\B(w)| - \varepsilon = c > 0$. Let $m$ be the maximal difference between a weight in $\A$ and a weight in $\B$. Since for every $\lambda>1$, $\sum_{i=0}^\infty (\frac{1}{\lambda^i}) = \frac{1}{1-\frac{1}{\lambda}} = \frac{\lambda}{\lambda-1}$, it follows that the difference between the values that $\A$ and $\B$ assign to any (finite or infinite) word is smaller than or equal to $\frac{m\lambda}{\lambda-1}$. Hence, the difference between the values of their runs on suffixes of $w$, starting at a position $p$, is smaller than or equal to $\frac{m\lambda}{(\lambda-1)\lambda^p} $. 

Now, since $\B$ $\varepsilon$-approximates $\A$ over finite words, it follows that they have optimal runs over every prefix of $w$, such that their difference is smaller than or equal to $\varepsilon$. Thus, after a long enough prefix, of length $p$ such that $\frac{m\lambda}{(\lambda-1)\lambda^p}  < c$, the difference between the values of $\A$'s and $\B$'s optimal runs on $w$ must be smaller than $c$, leading to a contradiction. 

A counter example for the converse is provided in Figure~\ref{fig:CounterExample}.
\end{proof}

Following Lemma~\ref{lem:FiniteApproxImpliesInfiniteApprox}, it is enough to prove the correctness of the constructions with respect to finite words, and the lower bounds with respect to infinite words.

Approximate determinization is useful for automata comparison, which is essential in formal verification, as well as for game solving, which is essential in synthesis. We briefly explain below how one can take advantage of approximate determinization for these purposes.

\Paragraph{Approximate comparison of automata} Consider two nondeterministic automata $\A$ and $\B$. One can approximately solve, with respect to a precision $\varepsilon > 0$, the question of whether for all words $w$, $\A(w) \geq \B(w)$. Now, what do we mean by ``approximately solve''? 

One may think that it allows to solve the question of whether for all words $w$, $(\B(w) - \A(w)) \leq \varepsilon$. However, this is not the case, as solving $(\B(w) - \A(w)) \leq \varepsilon$ is as difficult as solving $\B(w) \leq \A(w)$: Given $\lambda$-\NDAs $\A$ and $\B$, and some constant $\varepsilon$, one may construct an automaton $\B'$, such that for all words $w$, $\B'(w) = \B(w) + \varepsilon$. This is done by adding a constant weight $c$ to all weights in $\B$, where $c= \varepsilon \frac{\lambda-1}{\lambda}$. (The infinite discounted sum of $c$, with the discount factor $\lambda$, yields $\varepsilon$.) Then, $\B(w) \leq \A(w)$ if and only if $(\B'(w) - \A(w)) \leq \varepsilon$.

By ``approximately solve'' we mean that we can reduce the uncertainty area to be arbitrarily small: Given $\lambda$-\NDAs $\A$ and $\B$, and an arbitrary constant $\varepsilon > 0$,  we provide a ``yes'' or ``no'' answer, such that ``no'' means that $\A \not \geq \B$, and ``yes'' means that for all words $w$, $(\B(w) - \A(w)) \leq \varepsilon$. Note the there is an uncertainty area, in the size of $\varepsilon$, in the case of a ``yes'', meaning that for all words $w$, either $\A(w) \geq \B(w)$, or $\A(w)$ is almost as big as $\B(w)$, lacking an $\varepsilon$. 

We approximately solve, with respect to a precision $\varepsilon > 0$, the question of whether for all words $w$, $\A(w) \geq \B(w)$, as follows.
\begin{itemize}
\item We generate deterministic automata $\A'$ and $\B'$ that $\frac{\varepsilon}{4}$-approximate $\A$ and $\B$, respectively.
\item We construct an automaton $\C$, such that or all words $w$, $\C(w) = \B'(w)-\A'(w)$. This is done by taking $\C$ to be the product automaton of $\B'$ and $\A'$, where the weight of each transition is the weight from $\B$ minus the weight from $\A$. Note that for the nondeterministic automata $\A$ and $\B$, we cannot generate an automaton equivalent to $\B - \A$. (See Section~\ref{sec:Closure}.)
\item We compute the value $m = \sup\limits_w~ \C(w)$. Since $\C$ is deterministic, it can be solved using linear programming techniques. (See, for example, \cite{Andersson06}.)
\item If $m > \frac{\varepsilon}{2}$, we answer ``no'', and otherwise we answer ``yes''.
\end{itemize}

In the case that we answer ``no'', we know that there is a word $w$, such that $\B'(w) - \A'(w) > \frac{\varepsilon}{2}$. Since $\B'$ $\frac{\varepsilon}{4}$-approximates $\B$ and $\A'$ $\frac{\varepsilon}{4}$-approximates $\A$, it follows that $\B(w) - \A(w) > \frac{\varepsilon}{2} -  \frac{\varepsilon}{4} - \frac{\varepsilon}{4} =0$. Hence, $\A \not \geq \B$.

In the case that we answer ``yes'', we know that for all words $w$,  $\B'(w) - \A'(w) \leq \frac{\varepsilon}{2}$. Since $\B'$ $\frac{\varepsilon}{4}$-approximates $\B$ and $\A'$ $\frac{\varepsilon}{4}$-approximates $\A$, it follows that for all words $w$, $\B(w) - \A(w) \leq \frac{\varepsilon}{2} +  \frac{\varepsilon}{4} + \frac{\varepsilon}{4} =\varepsilon$. Hence, for all words $w$, $(\B(w) - \A(w)) \leq \varepsilon$.

The equivalence and universality problems (asking whether for all words $w$, $\A(w) = \B(w)$ and $\A(w) \leq 0$, respectively) can be approximately solved similarly, up to any desired precision.

\Paragraph{Approximate game solving} Consider a two-player game $\G$ whose value (winning condition) is given by means of a nondeterministic automaton $\A$. That is, $\G$ is a finite directed graph with edge weights, whose states are partitioned into two disjoint sets $S_1$ and $S_2$, belonging to $player_1$ and $player_2$, respectively. There is a distinguished initial state $s_0$ from which the plays of the game start. A \emph{play} $\rho$ is an infinite path in the graph, such that $player_1$ chooses the next state from a state in $S_1$, and $player_2$ chooses the next state from a state in $S_2$. A \emph{trace} $w_\rho$ of a play $\rho$ is the infinite sequence of weights generated by $\rho$. The \emph{value of a play} $\rho$ is defined to be $\A(w_\rho)$. The \emph{value of the game} is the value of a play $\rho$, in which both players follow their optimal strategy. (For more details on two-player games with quantitative objectives, see, for example, \cite{Andersson06} and \cite{CDHR10}.) 

For solving the game, meaning finding its value, one usually determinizes $\A$ into an automaton $\D$, takes the product of $\G$ and $\D$, and finds optimal strategies for the game $\G'=\G \times \D$. Now, in the case that $\D$ is not equivalent to $\A$, but $\frac{\varepsilon}{2}$-approximates it, the value of $\G'$ is guaranteed to be up to $\varepsilon$-apart from the value of $\G$. When defining that $player_1$ wins the game if the game's value is above some threshold, we can approximately solve the decision problem of whether $player_1$ wins the game, up to any desired precision, analogously to solving the automata-comparison problem, as elaborated above.

\subsection{Approximation by Unfolding}\label{sec:Unfolding}
We formalize below the naive approach of unfolding the automaton computations up to a sufficient level.

\Paragraph{The construction} Given an \NDA $\A = \tuple{\Sigma, Q, q_{in}, \delta, \gamma, \lambda}$
and a parameter $l \in \Nat$, we construct a \DDA $\D$ that is the depth-$l$ unfolding of $\A$.
We later fix the value of $l$ to obtain a \DDA that approximates $\A$ with a desired precision $\varepsilon$.

The \DDA is $\D = \tuple{\Sigma, Q', q'_{in}, \delta', \gamma', \lambda}$
where:
\begin{itemize}
\item $Q' = \Sigma^l$;  the set of words of length $l$.
\item $q'_{in} =$ the empty word.
\item $\delta' = \{(w,\sigma,w\cdot \sigma) \ST \abs{w} \leq l-1 \land \sigma \in \Sigma\}
\cup \{(w,\sigma,w) \ST \abs{w} = l \land \sigma \in \Sigma\}$.
\item For all $w \in \Sigma^{\leq l-1}$, and $\sigma \in \Sigma$, 
let $\gamma'(w,\sigma,w\cdot \sigma) =  (\A(w\cdot \sigma) - \A(w)) /  \lambda^{\abs{w}}$; 
for all $w \in \Sigma^{l}$, and $\sigma \in \Sigma$, let $\gamma'(w,\sigma,w) = \frac{v+V}{2}$
where $v$ and $V$ are the smallest and largest weights in $\A$, respectively.
\end{itemize}

\noindent The construction above yields an automaton whose state space might be doubly exponential in the representation of the discount factor.

\begin{thm}\label{thm:Unfolding}
Consider a precision $\varepsilon=2^{-p}$ and an \NDA $\A$ with a discount factor $\lambda = 1+ 2^{-k}$ and maximal weight difference of $m$. Then applying the unfolding construction on $\A$, for a precision $\varepsilon$, generates a \DDA $\D$ that $\varepsilon$-approximates $\A$ with up to $2^{\Theta(2^k (k+p+\log m) )}$ states.
\end{thm}
\begin{proof}
Let $l$ be the depth of $\A$'s unfolding that is used for generating $\D$. Then, for all words $w \in \Sigma^{\leq l}$, the automata $\A$ and $\D$
agree, by definition, on the value of $w$, that is $\A(w) = \D(w)$. 
For longer, or infinite, words $w \in \Sigma^{> l} \cup \Sigma^{\omega}$, we have:
$\D(w) = \A(w[0\dots l-1]) + \frac{v+V}{2}  \sum_{i=l}^{|w|} \frac{1}{\lambda^{i}}$.
As $\, v  \sum_{i=l}^{|w|} \frac{1}{\lambda^{i}} \leq \A(w) - \A(w[0\dots l-1]) \leq V  \sum_{i=l}^{|w|} \frac{1}{\lambda^{i}} \,$, we obtain the following:
$$ \abs{\A(w) - \D(w)} \leq \frac{V-v}{2}  \sum_{i=l}^{|w|} \frac{1}{\lambda^{i}} \leq \frac{V-v}{2}  \sum_{i=l}^{\infty} \frac{1}{\lambda^{i}} = \frac{V-v}{2}  \frac{1}{\lambda^l}\sum_{i=0}^{\infty} \frac{1}{\lambda^{i}} = \frac{m}{2  \lambda^{l-1}  (\lambda-1)}~,$$
\noindent where $m = V-v$ is the largest weight difference in $\A$.

Note that the above inequality is tight, in the sense that there is an automaton $\A$ and (an infinite) word $w$, such that $\abs{\D(w)-\A(w)} = \frac{m}{2  \lambda^{l-1}  (\lambda-1)}$.

In order to compute the minimal unfolding depth $l$ that guarantees a precision $\varepsilon = 2^{-p}$ when determinizing an automaton with a discount factor $\lambda = 1+ 2^{-k}$, we should solve the following inequality $\frac{m}{2  \lambda^{l-1}  (\lambda-1)}  = \frac{m 2^{k-1}}{\lambda^{l-1}} =  \frac{m 2^{k-1}}{(1+2^{-k})^{l-1}}\leq 2^{-p}$. 

Hence, $m 2^{k+p-1} \leq (1+2^{-k})^{l-1}$. Therefore, $(l-1)\log (1+2^{-k}) \geq k+p+\log (m) -1$, yielding that $l \geq \frac{k+p+\log (m) -1}{\log (1+2^{-k})}  +1$.

By Lemma~\ref{lem:HalfTime}, we have $\log (1+2^{-k})$ is linear in $2^{-k}$. Hence, $l \geq \Theta(2^k (k+ p+\log m) )$.

The unfolding construction of $\D$ generates up to $\Sigma^l$ states, implying that the deterministic automaton has up to $2^{\Theta(2^k (k+ p+\log m) )}$ states. 
\end{proof}

\subsection{Approximation by Gap Rounding}\label{sec:Approx}
As the unfolding approach, analyzed in Section~\ref{sec:Unfolding}, is doubly exponential in the discount factor, one may look for an alternative approach that is singly exponential in the discount factor, in the precision, and in the number of states of the original automaton. Indeed, we provide below such an approximation scheme, by generalizing the determinization procedure of Section~\ref{sec:Construction}. 

The main idea in Section~\ref{sec:Construction} is to extend the subset construction by keeping a recoverable-gap value to each element of the subset. Yet, it is shown in Section~\ref{sec:NonDeterminizability} that for every non-integral rational factor the construction might not terminate. The problem with non-integral factors is that the recoverable-gaps might be arbitrarily close to each other, implying infinitely many gaps within the maximal bound of recoverable gaps. 

Our approximation scheme generalizes the determinization procedure of Section~\ref{sec:Construction} by rounding the stored gaps to a fixed resolution. Since there is a bound on the maximal value of a recoverable gap, the fixed resolution guarantees the procedure's termination. The question is, however, how an unbounded number of gap rounding allows for the required precision. The key observation is that the rounding is also discounted along the computation.  For a  $\lambda$-\NDA, where $\lambda=1+2^{-k}$,  and a precision $\varepsilon = 2^{-p}$, we show that a resolution of $2^{-(p+k-1)}$ is sufficient. For an \NDA whose maximal weight difference is $m$, the maximal recoverable gap is below $m2^{k+1}$.
Hence, for an \NDA with $n$ states, the resulting \DDA would have up to $2^{n(p+2k+\log m)}$ states. 

The construction is formalized below, and illustrated with an example in Figure~\ref{fig:ApproxDet}.

\Paragraph{The construction}
Consider a discount factor $\lambda=1+2^{-k}$, with $k>0$, and an \NDA $\A = \tuple{\Sigma, Q=\tuple{q_1,\ldots,q_n}, q_{in}, \delta, \gamma, \lambda}$, in which the maximal difference between the weights is $m$. For simplicity, we extend $\gamma$ with $\gamma(\tuple{q_i,\sigma,q_j})=\infty$ for every $\tuple{q_i,\sigma,q_j}\not\in\delta$. Note that our discounted-sum automata do not have infinite weights; it is only used as an internal element of the construction.

For a precision  $\varepsilon=2^{-p}$, with $p>0$, we construct a \DDA $\D = \tuple{\Sigma, Q', q'_{in}, \delta', \gamma', \lambda}$ that $\varepsilon$-approximates $\A$. 
We first define the set $G = \{ i 2^{-(p+k-1)} \ST i\in\Nat  \mbox{ and } i \leq m2^{p+2k} \} \cup \{\infty \}$ of recoverable-gaps. The $\infty$ element denotes a non-recoverable gap, and behaves as the standard infinity element in the arithmetic operations that we will be using. 

A state of $\D$ extends the standard subset construction by assigning a recoverable gap to each state of $\A$. That is, $Q'= \{ \tuple{g_1, \ldots, g_n} \ST  \mbox{ for every } 1 \leq h \leq n, g_h\in G \}$.

The initial state of $\D$ is $q'_{in} = \tuple{g_1,\ldots,g_n}$, where for every $1 \leq i \leq n$, $g_i=0$ if $q_i = q_{in}$ and $g_i=\infty$ otherwise.

For bounding the number of possible states, the gaps are rounded to a resolution of $2^{-(p+k-1)}$. Formally, for every number $x\geq 0$, we define $\Round(x) =  i 2^{-(p+k-1)}$, such that $i\in\Nat$ and for every $j\in\Nat, |x-i 2^{-(p+k-1)}| \leq |x-j 2^{-(p+k-1)}|$. 

For every state $q'=\tuple{g_1,\ldots,g_n}\in Q'$ and letter $\sigma\in\Sigma$, we define the transition function $\delta(q',\sigma)=q''=\tuple{x_1,\ldots,x_n}$, and the weight function $\gamma(\tuple{q',\sigma,q''})=c$ as follows.
\begin{itemize}
\item For every $1 \leq h \leq n$, we set $c_h := \min \{g_j + \gamma(\tuple{q_j, \sigma, q_h}) \ST 1 \leq j \leq n  \}$
\item $c: = \min\limits_{1 \leq h \leq n}(c_h)$
\item For every $1 \leq l \leq n$, $x_l := \Round(\lambda(c_h-c))$. If $x_l \geq m2^{k+1}$ then $x_l := \infty$.\qed
\end{itemize}

\begin{figure}[htb]
% Exported from Xfig in a 45% ratio
\centering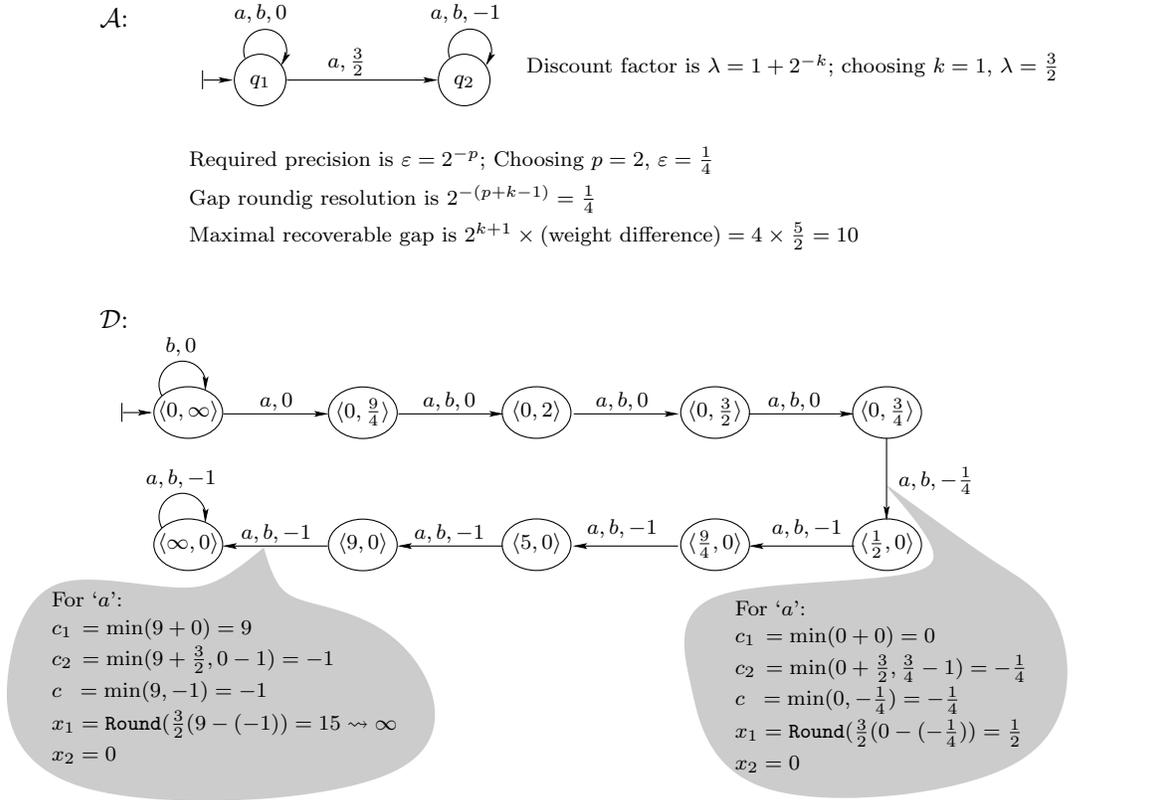 \caption{Determinizing the \NDA $\A$ approximately into the \DDA $\D$. The gray bubbles detail some of the intermediate calculations of the approximate-determinization construction.}\label{fig:ApproxDet}
\end{figure}

The correctness and the state complexity of the construction is formalized below.
\begin{thm}\label{thm:GapRounding}
Consider a discount factor $\lambda=1+2^{-k}$ and a precision $\varepsilon=2^{-p}$, for some positive numbers $p$ and $k$. Then for every $\lambda$-\NDA $\A$ with $n$ states and weight difference of up to $m$, there is a $\lambda$-\DDA that $\varepsilon$-approximates $\A$ with up to $2^{n(p+2k+\log m)}$ states. The automata $\A$ and $\D$ may operate over finite words as well as over infinite words.
\end{thm}
\begin{proof}
For every such \NDA $\A$, we construct the \DDA $\D$ as defined in the construction above. 

We show the correctness for finite words, implying, by Lemma~\ref{lem:FiniteApproxImpliesInfiniteApprox}, also the correctness for infinite words. 
We start by proving the claim with respect to an \emph{infinite-state} automaton $\D'$ that is constructed as above, except for not changing any gap to $\infty$. That is, the state space of $\D'$ is $\{ \tuple{g_1,\ldots,g_n} \ST \mbox{ for every } 1 \leq h \leq n,~ g_h = i2^{-(p+k-1)} \mbox{ for some } i \in \Nat \}$. Afterwards, we shall argue that changing all gaps that exceed $m^{2k+1}$ to $\infty$ does not harm the correctness.

We use the following notations: upon reading a word $w$, the automaton $\D'$ yields the sequence $c_1, c_2, \ldots, c_{|w|}$ of weights, and reaches a state $\tuple{g_{1,w}, \ldots, g_{n,w}}$. We denote half the gap resolution, namely $2^{-(p+k)}$, by $r$.
Intuitively, $\frac{g_{h,w}}{\lambda^{|w|}} + \sum_{i=1}^{|w|}\frac{c_i}{\lambda^{i-1}}$ stands for the approximated cost of reaching the state $q_h$ upon reading the word $w$. We define for every word $w$ and  $1 \leq h \leq n$, the ``mistake'' in the approximated cost by
$M(h,w) = \frac{g_{h,w}}{\lambda^{|w|}} + \sum_{i=1}^{|w|}\frac{c_i}{\lambda^{i-1}} -\Cost(q_{h},w)$,
 and show by induction on the length of  $w$ that  $ | M(h,w)  | \leq \sum_{i=1}^{|w|}\frac{r}{\lambda^i}$.
 
The assumptions obviously hold for the initial step, where $w$ is the empty word and all values are $0$. 
As for the induction step, we assume they hold for $w$ and show that for every $\sigma\in\Sigma$, they hold for $w\con\sigma$.

We first handle the case that $M(x,w\con\sigma) \geq 0$. 

Recall that for every $1 \leq x \leq n$, the actual cost of reaching $q_x$ is $\Cost(q_x,w\con\sigma) = \Cost(q_{h},w) + \frac{\gamma(\tuple{q_{h}, \sigma, q_x})}{\lambda^{|w|}}$, for some $1 \leq h \leq n$. By the construction, we have $g_{x,w\con\sigma} \leq \lambda (g_{h',w} + \gamma(\tuple{q_{h'}, \sigma, q_x}) - c_{|w|+1}) + r$, for some $1 \leq h' \leq n$. (The state $h'$ is the ``choice'' of the construction for the ``best'' state to continue with for reaching the state $x$ upon reading the word $w\con\sigma$.) The states $h$ and $h'$ need not be the same, yet, by the construction, $g_{h',w} + \gamma(\tuple{q_{h'}, \sigma, q_x}) \leq g_{h,w} + \gamma(\tuple{q_{h}, \sigma, q_x})$. Thus, $g_{x,w\con\sigma} \leq \lambda (g_{h,w} + \gamma(\tuple{q_{h}, \sigma, q_x}) - c_{|w|+1}) + r$.

Therefore, 
\begin{eqnarray*}
M(x,w\con\sigma) & = &  \frac{g_{x,w\con\sigma}}{\lambda^{|w|+1}} + \sum_{i=1}^{|w|+1}\frac{c_i}{\lambda^{i-1}} - \Cost(q_x,w\con\sigma)   \leq \\
&\leq& \frac{\lambda (g_{h,w} + \gamma(\tuple{q_{h}, \sigma, q_x}) - c_{|w|+1})+r}{\lambda^{|w|+1}} + \sum_{i=1}^{|w|+1}\frac{c_i}{\lambda^{i-1}} -\\
&&
- (\Cost(q_{h},w) + \frac{\gamma(\tuple{q_{h}, \sigma, q_x})}{\lambda^{|w|}}) = \\
& = &  \frac{g_{h,w}}{\lambda^{|w|}} + \frac{r}{\lambda^{|w|+1}} + \sum_{i=1}^{|w|}\frac{c_i}{\lambda^{i-1}} - \Cost(q_{h},w)  = \\
& = &  M(h,w) + \frac{r}{\lambda^{|w|+1}} \leq \sum_{i=1}^{|w|}\frac{r}{\lambda^i} +  \frac{r}{\lambda^{|w|+1}} = \sum_{i=1}^{|w|+1}\frac{r}{\lambda^i}~.
\end{eqnarray*}

We continue with the second case, where $M(x,w\con\sigma) \leq 0$. By the construction, we have $g_{x,w\con\sigma} \geq \lambda (g_{h',w} + \gamma(\tuple{q_{h'}, \sigma, q_x}) - c_{|w|+1}) - r$, for some $1 \leq h' \leq n$. (As in the previous case, the state $h'$ stands for the ``choice'' of the construction for the ``best'' state to continue with for reaching the state $x$ upon reading the word $w\con\sigma$.)

Then,
\begin{eqnarray*}
M(x,w\con\sigma) & = &  \frac{g_{x,w\con\sigma}}{\lambda^{|w|+1}} + \sum_{i=1}^{|w|+1}\frac{c_i}{\lambda^{i-1}} - \Cost(q_x,w\con\sigma)   \geq \\
& \geq &  \frac{\lambda (g_{h',w} + \gamma(\tuple{q_{h'}, \sigma, q_x}) - c_{|w|+1})-r}{\lambda^{|w|+1}} + \sum_{i=1}^{|w|+1}\frac{c_i}{\lambda^{i-1}} - \\
&& - (\Cost(q_{h},w) + \frac{\gamma(\tuple{q_{h}, \sigma, q_x})}{\lambda^{|w|}})~ .
\end{eqnarray*}

Since $\Cost(q_x, w\con\sigma) = \Cost(q_{h},w) + \frac{\gamma(\tuple{q_{h}, \sigma, q_x})}{\lambda^{|w|}}$, it follows that $ \Cost(q_{h},w) + \frac{\gamma(\tuple{q_{h}, \sigma, q_x})}{\lambda^{|w|}} \leq \Cost(q_{h'},w) + \frac{\gamma(\tuple{q_{h'}, \sigma, q_x})}{\lambda^{|w|}}$. (The accurate best path to reach the state $x$ upon reading the word $w\con\sigma$ goes through the state $h$, and not through the state $h'$.) Hence,  

\begin{eqnarray*}
M(x,w\con\sigma) & \geq & \frac{\lambda (g_{h',w} + \gamma(\tuple{q_{h'}, \sigma, q_x}) - c_{|w|+1}) - r}{\lambda^{|w|+1}} + \sum_{i=1}^{|w|+1}\frac{c_i}{\lambda^{i-1}} - \\
&& - (\Cost(q_{h'},w) + \frac{\gamma(\tuple{q_{h'}, \sigma, q_x})}{\lambda^{|w|}}) = \\
& = &   \frac{g_{h',w}}{\lambda^{|w|}} - \frac{r}{\lambda^{|w|+1}} + \sum_{i=1}^{|w|}\frac{c_i}{\lambda^{i-1}} - \Cost(q_{h'},w)  = \\
& = &  M(h',w) - \frac{r}{\lambda^{|w|+1}} \geq -\sum_{i=1}^{|w|}\frac{r}{\lambda^i} -  \frac{r}{\lambda^{|w|+1}} = - \sum_{i=1}^{|w|+1}\frac{r}{\lambda^i}~.
\end{eqnarray*}

We can now show that $|\D'(w)-\A(w)|\leq \varepsilon$. Indeed, upon reading $w$, $\D'$ reaches some state $q'$ in which $g_{h,w}=0$, for some $1 \leq h \leq n$. Thus,  $|\sum_{i=1}^{|w|}\frac{c_i}{\lambda^{i-1}} -\Cost(q_{h},w) | \leq \varepsilon$. Assume that $\A(w)>\D'(w)$.  Then, since $\A(w) \leq \Cost(q_{h},w)$, it follows that $0 \leq \A(w)-\D'(w) \leq \Cost(q_{h},w) - \D'(w) \leq \varepsilon$.  Analogously, assume that $\D'(w)>\A(w)$. Then, since $\A(w) \leq \Cost(q_{h},w)$, it follows that $0 \leq \D'(w)-\A(w) \leq \D'(w) - \Cost(q_{h},w) \leq \varepsilon$. 

It is left to show that $|\D(w)-\A(w)|\leq \varepsilon$. The only difference between the construction of $\D$ and of $\D'$ is that the former changes all gaps above $m(2^{k+1})$ to $\infty$. Upon reading $w$, $\D'$ ends in some state $q'$ in which $g_{h,w}=0$, for some $1 \leq h \leq n$. By the construction of $\D'$, there is a sequence of gaps $g_1, \ldots, g_{|w|}=g_{h,w}=0$, such that for every $i$, $g_{i+1}  \geq \lambda (g_i + x) - r$, where $|x| \leq m$. We claim that $\D$ also contains this sequence of gaps. Indeed, assume, by contradiction, that $g_i \geq m(2^{k+1})$, for some $i<|w|$. Then, $g_{i+1}  \geq \lambda (g_i -m - r) =  (1+2^{-k}) (g_i -m - r) = g_i -m - r + (2^{-k}) (g_i -m - r) \geq g_i -m - r + (2^{-k}) ( m(2^{k+1}) -m - r) = g_i -m - r + 2m -(2^{-k}) ( m+r) > g_i$. Hence, the sequence of gaps is growing from position $i$ onwards, contradicting the assumption that the last gap in the sequence is $0$.
\end{proof}

\subsection{Lower Bounds}
The upper bound described in Section~\ref{sec:Approx}, for determinizing an \NDA $\A$ approximately, exponentially depends on three parameters: $n$, denoting the number of states in $\A$; $k$, representing the proximity of $\A$'s discount factor to $1$; and $p$, representing the precision. We show below that exponential dependency on these three factors is unavoidable.

For showing dependency on the precision ($p$), as well as on the discount factor ($k$), one can fix the other two parameters and show exponential dependency on the varying parameter. This is formalized in Theorems \ref{thm:LbPrecision} and \ref{thm:LbDiscountFactor}. 

As for the number of states ($n$), there is no absolute dependency -- one may approximate the non-deterministic automaton via the unfolding approach (Section~\ref{sec:Unfolding}), having no dependency on $n$. Yet, the trade off is a double-exponential dependency on $k$. Thus, one may check whether there is an exponential dependency on $n$, once $p$ and $k$ are fixed, and $n$ remains below $O(2^k)$. This is indeed the case, as shown in Theorem~\ref{thm:LbStateNumber}.  

Intuitively, if a state of an \NDA has two recoverable gaps that are different enough, over two words that are short enough, then recovering them by the same suffix would yield two values that are also different enough. Two such words must lead to two different states in a deterministic automaton that properly approximates the \NDA. Hence, the challenge is to figure out an \NDA whose states have as many such different recoverable gaps as possible.
Each of the three lower bounds brings a different challenge, depending on the parameter that is not fixed (precision, discount factor, and number of states). A delicate analysis of the recoverable gaps in the automata of Figures ~\ref{fig:LbPrecision}--\ref{fig:LbStateNumber}, provides the proofs of Theorems~\ref{thm:LbPrecision}--\ref{thm:LbStateNumber}, respectively.

For showing dependency on the precision, we start with a lemma, analyzing the recoverable gaps of the \NDA on which we will show the lower bound.
\begin{figure}
% Exported from Xfig in a 50% ratio
\centering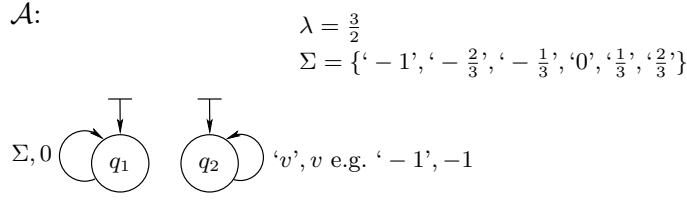 \caption{Every \DDA that $2^{-p}$-approximates the \NDA $\A$ has at least $2^{\lwhole {\frac{p-2}{log 3}}}$ states.} \label{fig:LbPrecision}
\end{figure}

\begin{lem}\label{lem:2OverlGaps}
Consider the \NDA described in Figure~\ref{fig:LbPrecision}. Then, for every $i,l\in\Nat$, where $i \leq 2^l$, there is a word $u_{l,i}\in\Sigma^l$ such that $q_2$ has the recoverable gap of $\frac{i}{2^l}$ over $u_{l,i}$.
\end{lem}
\begin{proof}
We prove the claim by induction on $l$. For $l=1$, we have the words $u_{1,0}=\letter{0}, u_{1,1}=\letter{\frac{1}{3}}$, and $u_{1,2}=\letter{\frac{2}{3}}$, for which $q_2$'s gaps are $\frac{3}{2} \times 0 = \frac{0}{2}$, $\frac{3}{2} \times \frac{1}{3} = \frac{1}{2}$, and $\frac{3}{2} \times \frac{2}{3} = \frac{2}{2}$, respectively. For the induction step, consider a number $j\leq 2^{l+1}$, and let $r=j \mod 3$. 

When $r=0$, we have the word $u_{l+1,j} = u_{l,j/3}\lcon\letter{0}$, as $\Gap(q_2,u_{l+1,j})=\frac{3}{2}(\Gap(q_2,u_{l,j/3}) - 0) = \frac{3}{2}\frac{j}{3 \times 2^l} = \frac{j}{2^{l+1}}$. 

When $r\neq0$, we check if $(2^l \mod 3)$ is different from $r$. If it is, we have the word $u_{l+1,j} = u_{l,(j+2^l)/3}\lcon\letter{-\frac{1}{3}}$, as $\Gap(q_2,u_{l+1,j})=\frac{3}{2}(\frac{j+2^l}{3\times 2^l} - \frac{1}{3})= \frac{j}{2^{l+1}}$. Otherwise, $(2^l \mod 3 = r)$ and $(2^{l+1} \mod 3) \neq r$. Then, in case that $j\leq 2^l$, we have the word $u_{l+1,j} = u_{l,(j+2^{l+1})/3}\lcon\letter{-\frac{2}{3}}$, as $\Gap(q_2,u_{l+1,j})=\frac{3}{2}(\frac{j+2^{l+1}}{3 \times 2^l} - \frac{2}{3})= \frac{j}{2^{l+1}}$, and in the case that $2^l<j\leq 2^{l+1}$, we have the word $u_{l+1,j} = u_{l,(j-2^l)/3}\lcon\letter{\frac{1}{3}}$, as $\Gap(q_2,u_{l+1,j})=\frac{3}{2}(\frac{j-2^{l}}{3 \times 2^l} + \frac{1}{3})= \frac{j}{2^{l+1}}$.

All these gaps of $q_2$ do not exceed $1$, thus they can be recovered by adding $\letter{-1}$ at the end of the words.
\end{proof}

We continue with the lower bound with respect to the precision.

\begin{thm}\label{thm:LbPrecision}
There is an \NDA $\A$ with two states, such that for every $\varepsilon=2^{-p}>0$, every \DDA (with any discount factor) that $\varepsilon$-approximates $\A$ has at least $2^{\lwhole {\frac{p-2}{\log 3}}}$ states.
\end{thm}
\begin{proof}
Consider the \NDA $\A$ described in Figure~\ref{fig:LbPrecision}, and let $\D$ be a \DDA (with any discount factor) that $\varepsilon$-approximates $\A$, where $\varepsilon=2^{-p}$. We claim that $\D$ has at least $2^{\lwhole{\frac{p-2}{\log 3}}}$ states.

Intuitively, if $q_2$ has different recoverable gaps over two short enough words, then recovering them would yield two values that are more than $\varepsilon$ apart. Two such  words must lead to two different states in a deterministic automaton that $\varepsilon$-approximates $\A$.

Formally, let $l=\lwhole {\frac{p-2}{\log 3}}$. By Lemma~\ref{lem:2OverlGaps}, for every $i\leq 2^l$, there is a word $u_{l,i}$, such that $\Gap(q_2, u_{l,i}) = \frac{i}{2^l}$. We show below that for every $i,j\leq2^l$, such that $i\neq j$, $\D$ reaches two different states upon reading $u_{l,i}$ and $u_{l,j}$, which implies that $\D$ has at least $2^{(\lwhole{\frac{p-2}{\log 3}})}$ states. 

Assume, by contradiction, two different words, $u=u_{l,i}$ and $u'=u_{l,j}$, such that $\D$ reaches the same state $s$ upon reading them. Let $w=u\lcon\letter{0}^\omega$ and $w'=u'\lcon\letter{0}^\omega$. Since $\A(w)=\A(w')=0$ and $\D$ $2^{-p}$-approximates $\A$, it follows that $|\D(w)-\D(w')|\leq 2\times 2^{-p}$. By the determinism of $\D$, we have $\D(w) = \D(u) + \frac{\D^s(\letter{0}^\omega)}{\lambda^l}$ and $\D(w') = \D(u') + \frac{\D^s(\letter{0}^\omega)}{\lambda^l}$. Hence, $|\D(u)-\D(u')|\leq 2\times2^{-p}=2^{(-p+1)}$.

Now, consider the words $z=u\lcon {-1}^\omega$ and $z'=u'\lcon\letter{-1}^\omega$. Since both $\Gap(q_2, u)$ and $\Gap(q_2, u')$ are recoverable by concatenating $\letter{-1}^\omega$, it follows that the best run of $\A$ over both $z$ and $z'$ goes via $q_2$. Thus,  $|\A(z) - \A(z')| =\frac{|\Gap(q_2, u) - \Gap(q_2, u')|}{\lambda^l} \geq \frac{1}{2^l} \times \frac{2^l}{3^l}= \frac{1}{3^l} > 2^{-(p-2)}$. Since $\D$ $2^{-p}$-approximates $\A$, it follows that $|\D(z) - \D(z')| > 2^{-(p-2)} - 2\times 2^{-p} = 2^{(-p+1)}$. By the determinism of $\D$, we have $|\D(u) - \D(u')| >  2^{-(p-1)}$. 

We have shown that $|\D(u) - \D(u')|$ is both smaller and bigger than $2^{(-p+1)}$, contradicting the assumption that $\D$ reaches the same state upon reading two different words $u_{l,i}$ and $u_{l,j}$. Hence, $\D$ has at least $2^{(\lwhole{\frac{p-2}{\log 3}})}$ states.
\end{proof}

We now turn to show the lower bound with respect to the discount factor.
\begin{figure}
% Exported from Xfig in a 50% ratio
\centering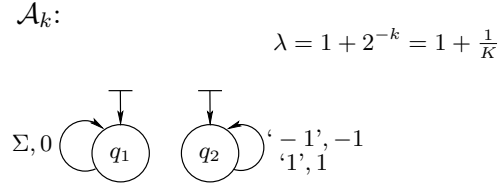 \caption{The family of \NDAs, such that every \DDA (with any discount factor) that $\frac{1}{8}$-approximates $\A_k$ has at least $2^{k-1}$ states.}\label{fig:LbDiscountFactor}
\end{figure}

\begin{thm}\label{thm:LbDiscountFactor}
There is a family of \NDAs, $\A_k$, for $k\geq 2$, with two states and discount factor $1+2^{-k}$ over an alphabet of two letters, such that every \DDA (with any discount factor) that $\frac{1}{8}$-approximates $\A_k$ has at least $2^{k-1}$ states.
\end{thm}
\begin{proof}
For convenience, we set $K=2^k$. For every $k\geq 2$, consider the \NDA $\A_k$, with discount factor $\lambda=1+2^{-k}=1+\frac{1}{K}$, described in Figure~\ref{fig:LbDiscountFactor}, and let $\D$ be a \DDA (with any discount factor) that $\frac{1}{8}$-approximates $\A_k$. We claim that $\D$ has at least $2^{k-1}$ states.

Intuitively, if $q_2$ has two recoverable gaps that are different enough, over two words that are short enough, then recovering them would yield two values that are also different enough. Two such words must lead to two different states in a deterministic automaton that $\frac{1}{8}$-approximates $\A_k$.

Formally, for every $1 \leq i \leq \frac{K}{2}$, consider the words $u_i=\letter{1}^i\lcon\letter{0}^{(K/2)-i}$. We claim that $q_2$ has recoverable gaps over all these words, and that the difference between the costs of reaching $q_2$ over each two words is at least $\half$. Indeed,  $\Gap(q_2, u_i) = \sum_{j=1}^{i}\lambda^j=\sum_{j=1}^{i}(1+\frac{1}{K})^j$. By Lemma~\ref{lem:HalfTime}, we have $(1+\frac{1}{K})^{K} < 3$, implying that $(1+\frac{1}{K})^{K/2} < \sqrt{3}$. Thus, $\Gap(q_2, u_i)  < \sum_{j=1}^{K/2}(1+\frac{1}{K})^{K/2} < (K/2) \sqrt{3} < K$. The maximal recoverable gap of $q_2$ is $\sum_{i=0}^\infty (\frac{1}{\lambda^i}) = \frac{\lambda}{\lambda-1}  = \frac{1+\frac{1}{K}}{\frac{1}{K}}=K+1$, implying that $q_2$ has a recoverable gap over $u_i$. As for the costs of reaching $q_2$ over different words, consider the words $u_i$ and $u_j$, where $i \neq j$. We have $|\Cost(q_2,u_i)-\Cost(u_j)| \geq \frac{1}{(1+\frac{1}{K})^{K/2}} > \frac{1}{\sqrt{3}} > \frac{1}{2}$.

We continue with analyzing the runs of $\D$ on these words. Assume, by contradiction, two different words, $u_i$ and $u_j$, such that $\D$ reaches the same state $s$ upon reading them. Let $w=u_i\lcon\letter{0}^\omega$ and $w'=u_j\lcon\letter{0}^\omega$. Since $\A(w)=\A(w')=0$ and $\D$ $\frac{1}{8}$-approximates $\A$, it follows that $|\D(w)-\D(w')|\leq 2\times \frac{1}{8} = \frac{1}{4}$. By the determinism of $\D$, we have $\D(w) = \D(u_i) + \frac{\D^s(\letter{0}^\omega)}{\lambda^l}$ and $\D(w') = \D(u_j) + \frac{\D^s(\letter{0}^\omega)}{\lambda^l}$. Hence, $|\D(u_i)-\D(u_j)|\leq \frac{1}{4}$.

Now, consider the words $z=u_i\lcon \letter{-1}^\omega$ and $z'=u_j\lcon\letter{-1}^\omega$. Since both $\Gap(q_2, u_i)$ and $\Gap(q_2, u_j)$ are recoverable by concatenating $\letter{-1}^\omega$, it follows that the best run of $\A$ over both $z$ and $z'$ goes via $q_2$. Thus,  $|\A(z) - \A(z')| = |\Cost(q_2,u_i)-\Cost(q_2,u_j)|  > \frac{1}{2}$. Since $\D$ $\frac{1}{8}$-approximates $\A$, it follows that $|\D(z) - \D(z')| > \frac{1}{2} - 2\times \frac{1}{8} = \frac{1}{4}$. By the determinism of $\D$, we have $|\D(u) - \D(u')| >  \frac{1}{4}$. 

We have shown that $|\D(u_i) - \D(u_j)|$ is both smaller and bigger than $\frac{1}{4}$, contradicting the assumption that $\D$ reaches the same state upon reading two different words $u_i$ and $u_j$. Hence, $\D$ has at least $2^{K/2}=2^{p-1}$ states.
\end{proof}

We conclude with the lower bound with respect to the number of states, which also depends on the discount factor.

\begin{figure}
% Exported from Xfig in a 50% ratio
\centering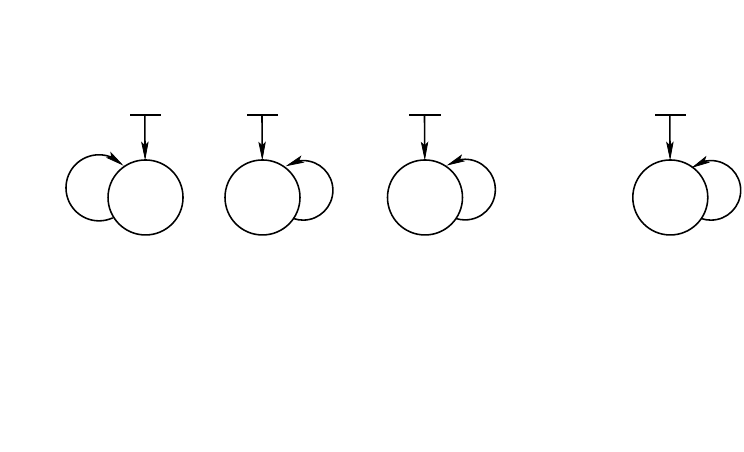 \caption{The family of \NDAs, such that every \DDA (with any discount factor) that $\frac{1}{12}$-approximates $\A_n$ has at least $2^n$ states.}\label{fig:LbStateNumber}
\end{figure}

\begin{thm}\label{thm:LbStateNumber}
For every $k\geq 3$, there is a family of \NDAs, $\A_n$, for $n\leq 2^k$, with $n+1$ states and discount factor $1+2^{-k}$ over an alphabet of size $2n+1$, such that every \DDA (with any discount factor) that $\frac{1}{12}$-approximates $\A_n$ has at least $2^{n}$ states.
\end{thm}
\begin{proof}
For convenience, we set $K=2^k$. For every $k\geq 2$ and $n\leq2^k$, consider the \NDA $\A_n$, with discount factor $\lambda=1+2^{-k}=1+\frac{1}{K}$, described in Figure~\ref{fig:LbStateNumber}, and let $\D$ be a \DDA (with any discount factor) that $\frac{1}{12}$-approximates $\A_n$. We claim that $\D$ has at least $2^{n}$ states.

Intuitively, we describe $2^n$ words of length $n$, encoding all binary combinations of $n$ bits. The cost of reaching a state $q_i$ over a word $u$ is above $\frac{1}{3}$ if the $i$th bit in $u$ is $1$, and zero otherwise. These words are recoverable for all states, implying that a deterministic automaton that $\frac{1}{12}$-approximates $\A$ should reach a different state upon reading each of the $2^n$ words.

Formally, for every binary word $b\in \{0,1\}^n$ of length $n$, we define the word $u_b\in\Sigma^n$, by setting the $i$th letter of $u_b$ to $\letter{1_i}$ if the $i$th letter of $b$ is $1$, and to $\letter{0}$ otherwise. We denote this set of $\Sigma^n$-words by $U$. 

For every $1\leq i \leq n$ and $u\in U$, we have $\Cost(q_i,u)=0$ if $u[i]=\letter{0}$ and $\frac{1}{\lambda^i}$ otherwise. By Lemma~\ref{lem:HalfTime}, we have $(1+\frac{1}{K})^{K} < 3$. Since, $i\leq n\leq K$, it follows that $\Cost(q_i,u)>\frac{1}{3}$ if $u[i]\neq\letter{0}$. In addition, we have that for every $1\leq i \leq n$ and $u\in U$, the gap of $q_i$ over $u$ is recoverable by concatenating $\letter{-1_i}^\omega$, as $\Gap(q_i,u)\leq \lambda^n \leq (1+\frac{1}{K})^{K} < 3$ and the value of concatenating $\letter{-1_i}^\omega$ from the $n$th position is $\frac{1}{\lambda^n} \sum_{j=1}^\infty \lambda^j \geq \frac{1}{3} (K+1) \geq \frac{1}{3}(2^3+1) = 3$.

We continue with analyzing the runs of $\D$ on these words. Assume, by contradiction, two different words, $u,u'\in U$, such that $\D$ reaches the same state $s$ upon reading them. Let $w=u\lcon\letter{0}^\omega$ and $w'=u'\lcon\letter{0}^\omega$. Since $\A(w)=\A(w')=0$ and $\D$ $\frac{1}{12}$-approximates $\A$, it follows that $|\D(w)-\D(w')|\leq 2\times \frac{1}{12} = \frac{1}{6}$. By the determinism of $\D$, we have $\D(w) = \D(u) + \frac{\D^s(\letter{0}^\omega)}{\lambda^l}$ and $\D(w') = \D(u') + \frac{\D^s(\letter{0}^\omega)}{\lambda^l}$. Hence, $|\D(u_i)-\D(u_j)|\leq \frac{1}{6}$.

Now, since $u\neq u'$, there is some index $1\leq i \leq n$, such that $u[i]=\letter{0}$ and $u'[i]=\letter{1_i}$, or vice versa. 
Consider the words $z=u\lcon (-1_i)^\omega$ and $z'=u'\lcon\letter{-1_i}^\omega$. One can observe that for every $0\leq j \leq n$, such that $j\neq i$, we have $\Cost(q_j,u)\geq 0$, $\Cost(q_j,u')\geq 0$, and $\A^{q_j}(\letter{-1}^\omega) = 0$. On the other hand, $\A(z) < 0$ and $\A(z') < 0$, going via $q_i$. Thus,  $|\A(z) - \A(z')| = |\Cost(q_i,u)-\Cost(q_i,u')|  > \frac{1}{3}$. Since $\D$ $\frac{1}{12}$-approximates $\A$, it follows that $|\D(z) - \D(z')| > \frac{1}{3} - 2\times \frac{1}{12} = \frac{1}{6}$. By the determinism of $\D$, we have $|\D(u) - \D(u')| >  \frac{1}{6}$. 

We have shown that $|\D(u) - \D(u')|$ is both smaller and bigger than $\frac{1}{6}$, contradicting the assumption that $\D$ reaches the same state upon reading two different words $u$ and $u'$. Hence, $\D$ has at least $2^n$ states.
\end{proof}

\section{Closure Properties}\label{sec:Closure}
Discounted-sum automata realize a function from words to numbers.
Hence, one may wish to consider their closure under arithmetic operations.
The operations are either between two automata, having the same discount factor, as addition and taking the minimum, or between an automaton and a scalar, as multiplication by a positive rational number $c$.

Formally, given automata $\A$ and $\B$, and a scalar $0\leq c\in\Rat$, we define
\begin{itemize}
\item $\C = \min (\A, \B)$ if for every word $w$, $\C(w) = \min (\A(w), \B(w))$.
\item $\C = \max (\A, \B)$ if for every word $w$, $\C(w) = \max (\A(w), \B(w))$.
\item $\C =  \A + \B$ if for every word $w$, $\C(w) = \A(w) + \B(w)$.
\item $\C =  \A - \B$ if for every word $w$, $\C(w) = \A(w) - \B(w)$.
\item $\C =  \A \cdot c $ if for every word $w$, $\C(w) = c \cdot \A(w)$.
\item $\C =  \A \cdot (-1)$ if for every word $w$, $\C(w) = - \A(w)$.
\end{itemize}

We consider the class of complete \NDAs, as well as two of its subclasses: \DDAs and integral \NDAs.

The closure properties, summarized in Table~\ref{tbl:Closure}, turn out to be the same for automata over finite words and over infinite-words. By arguments similar to those of Lemma~\ref{lem:FiniteEquivalenceImpliesInfiniteEquivalence}'s proof, it is enough to prove the positive results with respect to automata over finite words and the negative results with respect to automata over infinite words.

\begin{table}
    \begin{center}
        \begin{tabular}{|c||c|c|c|c|c|c|}
        \hline
	Class $\diagdown$ Operation & $~\min~$ & $~\max~$ & $~~~+~~~$ & $~~~-~~~$ & $\cdot c, c \geq 0$ & $\cdot(-1)$\\
        \hline \hline
	\NDAs & \NiceCheckMark & \NiceCross & \NiceCheckMark & \NiceCross & \NiceCheckMark & \NiceCross\\
        \hline
	\DDAs & \multicolumn{2}{c|}{\NiceCross}  & \multicolumn{4}{c|}{\NiceCheckMark}\\
        \hline
	Integral \NDAs& \multicolumn{6}{c|}{\NiceCheckMark}\\
        \hline
        \end{tabular}
    \end{center}
    \caption{Closure of discounted-sum automata under arithmetic operations.} \label{tbl:Closure}
\end{table}

Some of the positive results are straightforward, as follows.
\begin{itemize}
\item Nondeterministic: The automaton that provides the minimum between two automata is achieved by taking the union of the input automata, addition by taking the product of the input automata and adding the corresponding weights, and multiplication by a positive scalar $c$ is achieved by multiplying all weights by $c$. 
\item Deterministic: Addition/subtraction is achieved by taking the product of the input automata and adding/subtracting the corresponding weights. Multiplication by (positive or negative) scalar $c$ is achieved by multiplying all weights by $c$. 
\item  Discount factor $\in\Nat$: Since these automata can always be determinized, they obviously enjoy the closure properties of both the deterministic and non-deterministic classes.
\end{itemize}

\noindent All the negative results can be reduced to the $\max$ operation, as follows. Closure under subtraction implies closure under $(-1)$-multiplication, by subtracting the given automaton from a constant $0$ automaton. For nondeterministic automata, closure under $(-1)$-multiplication implies closure under the $\max$ operation, by multiplying the original automata by $(-1)$ and taking their minimum. As for deterministic automata, closure under the $\min$ and $\max$ operations are reducible to each other due to the closure under $(-1)$-multiplication. 

It is left to show the results with respect to the $\max$ operation. We start with the classes of deterministic and nondeterministic automata.

\begin{thm}\label{thm:MaxInclosure}
\NDAs and \DDAs are not closed under the $\max$ operation.
\end{thm}
\begin{proof}
We prove a stronger claim, showing that there are two \DDAs, $\A$ and $\B$, defined in Figure~\ref{fig:NoMax}, for which there is no \NDA equivalent to $\max(\A,\B)$. Intuitively, we show that the recoverable-gap between $\A$ and $\B$ can be arbitrarily small, and therefore, by pumping-arguments, an \NDA for $\max(\A,\B)$ cannot be of a finite size.

Assume, by contradiction, an \NDA $\C$ with $n$ states equivalent to $\max(\A,\B)$.
The value of $\A$ over every word is obviously $0$. Thus, for every infinite word $w$, $\C(w) = \B(w)$ if $\B(w)>0$ and $0$ otherwise.

For a finite word $u$, we shall refer to $\lambda^{|u|}\B(u)$ as the \emph{gap} of $\B$ over $u$, denoted $\Gap(\B,u)$. Intuitively, this gap stands for the weight that $\B$ should save over a suffix $z$ for having a negative value over the whole word. That is, $\B(uz) < 0$ if and only if $\B(z) < -\Gap(\B,u)$. Within this proof, $\lambda$ is fixed to $\frac{5}{2}$.

A key observation is that the gap of $\B$ can be arbitrarily small. Specifically, we show that for every natural numbers $k\geq3$ and  $j \leq \uwhole{\frac{2^k}{5}}$, there is a finite word $u_{j,k}$ such that $\Gap(\B,u_{j,k})=\frac{5j}{2^k}$. It goes by induction on $k$. For $k=3$, it holds with $u_{0,3}=$`$0$',  $u_{1,3}=$`$\frac{2}{5}$' `$-\frac{1}{2}$' `$-1$', and $u_{2,3}=$`$\frac{2}{5}$' `$-\frac{1}{2}$'. As for the induction step, consider a number $0\leq j \leq \uwhole{\frac{2^{k+1}}{5}}$. One may verify that multiplying $2^k$ by each of $1, \frac{1}{2},  \frac{1}{4}, \frac{1}{8}$, and $0$, provides a different reminder when divided by $5$. Hence, there is a number $v\in\{ -1, -\frac{1}{2},  -\frac{1}{4}, -\frac{1}{8}, 0\}$ and a natural number $j' \leq \uwhole{\frac{2^k}{5}}$ such that $j' = \frac{j - v 2^k}{5}$. Thus, we can have, by the induction assumption, the required word $u_{j,k+1}$, by $u_{j,k+1}=u_{j',k}$`$v$' , as $\Gap(\B,u_{j,k+1})=\frac{5}{2}(\Gap(\B,u_{j',k}) + v)  =\frac{5}{2}(\frac{5j'}{2^k} + v) = \frac{5}{2} (\frac{j-v2^k}{2^k}+v) =  \frac{5j}{2^{k+1}}$.

By the above observation, there is a finite word $u$, such that $\frac{1}{\lambda^{2n}} < \Gap(\B,u)<\frac{1}{\lambda^n}$. We define the infinite word $w=u $`$0$'$^n$ `$-1$'  `$0$'$^\omega$. Since a $0$-weighted letter multiplies the gap by $\lambda$, we get that $0 < \Gap(\B,u$`$0$'$^n) < 1$, and therefore $\B(w)<0$ and $\C(w)=0$.

Let $r$ be an optimal run of $\C$ on $w$. Since $\C$ has only $n$ states, there is a state $q$ of $\C$ and two positions $|u| \leq p_1 < p_2 < |u|+ n$, such that $r$ visits $q$ on both $p_1$ and $p_2$. Let $w_1$ and $w_2$ be the prefixes of $w$ of lengths $p_1$ and $p_2$, respectively. Let $z$ be the suffix of $w$ after $w_2$, that is $w=w_2 z$. Let $r_1$, $r_2$ and $r_z$ be the portions of $r$ on $w_1$, $w_2$ and $z$, respectively.

Let $v_1$ and $v_2$ be the values of $r_1$ and $r_2$, respectively, and define $g_1 = \lambda^{p1}v_1$ and $g_2 = \lambda^{p2}v_2$. Let $v_z$ be the value of a run equivalent to $r_z$. Since the value of $r$ is $0$, we have that $v_2 + \frac{v_z}{\lambda^{p_2}} = 0$, and therefore, $v_z = -g_2$.  

We shall reach a contradiction by showing that  $g_1 \not < g_2$, $g_1 \not > g_2$, and $g_1\neq g_2$. Indeed:
\begin{itemize}
\item If $g_1< g_2$ then there is a run $r'=r_1 r_z$ of $\C$ on the word $w'=w_1z$, whose value is $v_1 + \frac{v_z}{\lambda^{p_2}} = \frac{g_1 + v_z}{\lambda^{p_2}}$. However, since $g_1< g_2=-v_z$, it follows that the value of $\C$ on $w'$ is negative, which leads to a contradiction.
\item If $g_1> g_2$ then there is a negative-valued run of $\C$ on the word $w_1$`$0$'$^{2(p_2-p_1)}z$, analogously to the previous case.
\item If $g_1= g_2$ then there is a $0$-valued run of $\C$ on the word $w'=w_1$`$0$'$^{2n}z$, however $\B(w')>0$, leading to a contradiction.\qedhere
\end{itemize}
\end{proof}

\begin{figure}
% Exported from Xfig in a 50% ratio
\centering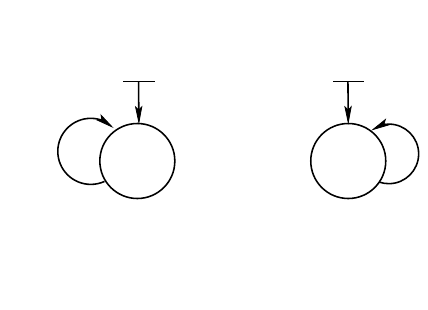 \caption{The \DDAs $\A$ and $\B$, for which there is no \NDA equivalent to $\max(\A,\B)$.}\label{fig:NoMax}
\end{figure}

\noindent We continue with the class of automata with an integral factor.
\begin{thm}\label{thm:MaxClosure}
For every $\lambda\in\Nat$, the class of $\lambda$-\NDAs is closed under the $\max$ operation.
\end{thm}
\begin{proof}
Consider a discount-factor $1<\lambda\in\Nat$ and two $\lambda$-\NDAs, $\A$ and $\B$. By Theorem~\ref{thm:Determinizability}, $\A$ and $\B$ can be determinized to equivalent $\lambda$-\DDAs. Thus, we may only consider deterministic automata. Since deterministic automata are closed under $(-1)$-multiplication, we may also consider the $\min$ operation rather than the $\max$ operation.

The construction of a \DDA $\C$ equivalent to $\min(\A, \B)$ is analogous to the determinization construction of Section~\ref{sec:Construction}, with the difference of extending automata-product rather than the subset-construction. Namely, we iteratively construct the product of $\A$ and $\B$, where a state of $\C$ contains a state of $\A$ and a state of $\B$, together with their recoverable-gaps. That is, for a state $p$ of $\A$ and a state $q$ of $\B$, a state $c$ of $\C$ is of the form $c=\pair{\pair{p,g_p}, \pair{q,g_q}}$. When $\A$ and $\B$ read a finite word $u$ and reach the states $p$ and $q$, respectively, we have that $g_p = \lambda^{|u|}( \A(u) - \min(\A(u), \B(u)))$ and $g_q = \lambda^{|u|}( \B(u) - \min(\A(u), \B(u)))$. Once a gap is too large, meaning that it is bigger than twice the maximal difference between a weight in $\A$ and a weight in $\B$, it is changed to $\infty$.

The termination and correctness proofs of the above construction are analogous to the proofs of Lemmas \ref{lem:Termination} and \ref{lem:DetCorrectness}.
\end{proof}

\section{Conclusion}
Recently, there has been a considerable effort to extend formal verification from the Boolean setting to a quantitative one. Automata theory plays a key role in formal verification, and therefore quantitative automata, such as limit-average automata and discounted-sum automata, play a central role in quantitative formal verification. However, of the basic automata questions underlying a verification task, namely, emptiness, universality, and inclusion, only emptiness is known to be solvable for these automata. The other questions are either undecidable, as is the case with limit-average automata, or not known to be decidable, as is the case with discounted-sum automata. 

We showed that discounted-sum automata with integral discount factors form a robust class, having algorithms for all the above questions, being closed under natural composition relations, such as $\min$, $\max$, addition and subtraction, and allowing for determinization. For discounted-sum automata with a nonintegral factor, we showed that they can be determinized approximately with respect to any required precision, which is not the case with other quantitative automata, such as sum, average, and limit-average automata. Hence, we find the class of discounted-sum automata a promising direction in the development of formal quantitative verification.

\section*{Acknowledgement}
We thank Laurent Doyen for great ideas and valuable help, and the anonymous reviewers for their very helpful comments and suggestions.

\bibliographystyle{alpha}
\bibliography{bib}

\newcommand{\etalchar}[1]{$^{#1}$}
\begin{thebibliography}{DDG{\etalchar{+}}10}

\bibitem[ABK11]{ABK11}
S.~Almagor, U.~Boker, and O.~Kupferman.
\newblock What's decidable about weighted automata?
\newblock In {\em ATVA}, volume 6996 of {\em LNCS}, pages 482--491, 2011.

\bibitem[AKL11]{AKL11}
B.~Aminof, O.~Kupferman, and R.~Lampert.
\newblock Rigorous approximated determinization of weighted automata.
\newblock In {\em Proc.\ of LICS}, pages 345--354, 2011.

\bibitem[And06]{Andersson06}
D.~Andersson.
\newblock An improved algorithm for discounted payoff games.
\newblock In {\em Proc.\ of ESSLLI Student Session}, pages 91--98, 2006.

\bibitem[BGW01]{BGW01}
A.~L. Buchsbaum, R.~Giancarlo, and J.~Westbrook.
\newblock An approximate determinization algorithm for weighted finite-state
  automata.
\newblock {\em Algorithmica}, 30(4):503--526, 2001.

\bibitem[BH11]{BH11}
U.~Boker and T.~A. Henzinger.
\newblock Determinizing discounted-sum automata.
\newblock In {\em Proc.\ of CSL}, volume~12 of {\em LIPIcs}, pages 82--96,
  2011.

\bibitem[BH12]{BH12}
U.~Boker and T.~A. Henzinger.
\newblock Approximate determinization of quantitative automata.
\newblock In {\em Proc.\ of FSTTCS}, volume~18 of {\em LIPIcs}, pages 362--373,
  2012.

\bibitem[CDH09]{AlternatingWeightedAutomata}
K.~Chatterjee, L.~Doyen, and T.~A. Henzinger.
\newblock Alternating weighted automata.
\newblock In {\em Proc.\ of FCT}, volume 5699 of {\em LNCS}, pages 3--13, 2009.

\bibitem[CDH10a]{ExpressivenessQuantitativeLanguages}
K.~Chatterjee, L.~Doyen, and T.~A. Henzinger.
\newblock Expressiveness and closure properties for quantitative languages.
\newblock {\em Logical Methods in Computer Science}, 6(3), 2010.

\bibitem[CDH10b]{CDH10}
K.~Chatterjee, L.~Doyen, and T.~A. Henzinger.
\newblock Quantitative languages.
\newblock {\em ACM Trans. Comput. Log.}, 11(4), 2010.

\bibitem[CDHR10]{CDHR10}
K.~Chatterjee, L.~Doyen, T.~A. Henzinger, and J.~F. Raskin.
\newblock Generalized mean-payoff and energy games.
\newblock In {\em FSTTCS}, volume~8 of {\em LIPIcs}, pages 505--516, 2010.

\bibitem[dAHM03]{DiscountingInSystems}
L.~de~Alfaro, T.~A. Henzinger, and R.~Majumdar.
\newblock Discounting the future in systems theory.
\newblock In {\em Proc.\ of ICALP}, volume 2719 of {\em LNCS}, pages
  1022--1037, 2003.

\bibitem[DDG{\etalchar{+}}10]{ImperfectInformation}
A.~Degorre, L.~Doyen, R.~Gentilini, J.~F. Raskin, and S.~Torunczyk.
\newblock Energy and mean-payoff games with imperfect information.
\newblock In {\em Proc.\ of CSL}, volume 6247 of {\em LNCS}, pages 260--274,
  2010.

\bibitem[DK06]{Skew}
M.~Droste and D.~Kuske.
\newblock Skew and infinitary formal power series.
\newblock {\em Theor. Comput. Sci.}, 366(3):199--227, 2006.

\bibitem[DKV09]{WeightedHandbook}
M.~Droste, W.~Kuich, and H.~Vogler.
\newblock {\em Handbook of Weighted Automata}.
\newblock Springer Publishing Company, Incorporated, 2009.

\bibitem[FGR12]{FGR12}
E.~Filiot, R.~Gentilini, and J.~F. Raskin.
\newblock Quantitative languages defined by functional automata.
\newblock In {\em CONCUR}, volume 7454 of {\em LNCS}, pages 132--146, 2012.

\bibitem[GZ07]{DiscountedMarkov}
H.~Gimbert and W.~Zielonka.
\newblock Limits of multi-discounted markov decision processes.
\newblock In {\em Proc.\ of LICS}, pages 89--98, 2007.

\bibitem[Moh97]{MohriAlgorithms}
M.~Mohri.
\newblock Finite-state transducers in language and speech processing.
\newblock {\em Computational Linguistics}, 23:269--311, 1997.

\bibitem[MTZ10]{DiscountedDeterministicMarkov}
O.~Madani, M.~Thorup, and U.~Zwick.
\newblock Discounted deterministic markov decision processes and discounted
  all-pairs shortest paths.
\newblock {\em ACM Transactions on Algorithms}, 6(2), 2010.

\bibitem[ZP96]{ZP96}
U.~Zwick and M.S. Paterson.
\newblock The complexity of mean payoff games on graphs.
\newblock {\em Theoretical Computer Science}, 158:343--359, 1996.

\end{thebibliography}

\end{document}